\documentclass[runningheads]{llncs}

 \newcommand{\ExtendedVersion}[1]{#1}\newcommand{\PaperVersion}[1]{}

\usepackage[T1]{fontenc}
\usepackage{subcaption}
\usepackage{amsmath, amssymb}					% packages to get the fonts, symbols used in most math
\usepackage[english]{babel}
\usepackage{hyperref}
\usepackage{multirow}
\usepackage{latexsym}
\usepackage{listings} 
\usepackage{booktabs}
\usepackage{array}
\usepackage{wrapfig}
\usepackage[inline]{enumitem}
    \setlist[enumerate,2]{ref={(\alph*)}} % cross-references contain only the sub-item number

\usepackage{graphicx}
\usepackage[dvipsnames,svgnames, table]{xcolor} % text color
\DeclareFontFamily{U}{mathx}{\hyphenchar\font45}
\DeclareFontShape{U}{mathx}{m}{n}{
      <5> <6> <7> <8> <9> <10>
      <10.95> <12> <14.4> <17.28> <20.74> <24.88>
      mathx10
      }{}
\DeclareSymbolFont{mathx}{U}{mathx}{m}{n}
\DeclareMathSymbol{\bigtimes}{1}{mathx}{"91}

\usepackage{etoolbox} % For \patchcmd
\usepackage{orcidlink} % to use \orcidlink{...} instead of \orcidID{...}

\usepackage[linesnumbered, vlined, ruled]{algorithm2e}
	\SetAlCapHSkip{0px} % to have the caption of the algorithms start at the left margin
	\SetNlSty{}{}{}     % to avoid bold-font line numbers

\makeatletter
\patchcmd{\algocf@makecaption@ruled}{\hsize}{\textwidth}{}{} % Caption to stretch to full text width
\patchcmd{\@algocf@start}{-1.5em}{0em}{}{} % For comments in algorithms to strethc to the right margin
\makeatother

\makeatletter
\renewcommand{\paragraph}{%
  \@startsection{paragraph}{4}%
  {\z@}{1.5ex \@plus 1ex \@minus .2ex}{-1em}%
  {\normalfont\normalsize\bfseries}%
}
\makeatother

\newcounter{querycnt}

\definecolor{eclipseStrings}{RGB}{42,0.0,255}
\definecolor{eclipseKeywords}{RGB}{127,0,85}
\definecolor{darkgreen}{RGB}{0,128,21}
\colorlet{numb}{magenta!60!black}
\usepackage[normalem]{ulem} % wavy underlines
\makeatletter
\def\uwave{\bgroup \markoverwith{\lower3.5\p@\hbox{\sixly \textcolor{red}{\char58}}}\ULon}
\font\sixly=lasy6 % does not re-load if already loaded, so no memory problem.
\makeatother

\usepackage{xargs}                      % Use more than one optional parameter in a new commands
\usepackage[colorinlistoftodos,prependcaption,textsize=tiny]{todonotes}
\newcommandx{\unsure}[2][1=]{\todo[linecolor=red,backgroundcolor=red!25,bordercolor=red,#1]{#2}}
\newcommandx{\change}[2][1=]{\todo[linecolor=blue,backgroundcolor=blue!25,bordercolor=blue,#1]{#2}}
\newcommandx{\info}[2][1=]{\todo[linecolor=OliveGreen,backgroundcolor=OliveGreen!25,bordercolor=OliveGreen,#1]{#2}}
\newcommandx{\improvement}[2][1=]{\todo[linecolor=Plum,backgroundcolor=Plum!25,bordercolor=Plum,#1]{#2}}
\newcommandx{\thiswillnotshow}[2][1=]{\todo[disable,#1]{#2}}
\newcommandx{\todoi}[2][1=]{\todo[inline,size=\small,linecolor=Plum,backgroundcolor=Periwinkle!25,#1]{#2}}
\newcommandx{\todoCite}[2][1=]{\todo[linecolor=Plum,backgroundcolor=Periwinkle!25,bordercolor=Plum,#1]{#2}}
\newcommandx{\todoDef}[2][1=]{\todo[linecolor=Plum,backgroundcolor=Periwinkle!25,bordercolor=Plum,#1]{Def.: #2}}

\makeatletter
\font\uwavefont=lasyb10 scaled 700
\def\spelling{\bgroup\markoverwith{\lower3.5\p@\hbox{\uwavefont\textcolor{Red}{\char58}}}\ULon}
\def\grammar{\bgroup\markoverwith{\lower3.5\p@\hbox{\uwavefont\textcolor{LimeGreen}{\char58}}}\ULon}
\def\phrasing{\bgroup\markoverwith{\lower3.5\p@\hbox{\uwavefont\textcolor{blue}{\char58}}}\ULon}

\newcommand\remove{\bgroup\markoverwith{\textcolor{red}{\rule[0.5ex]{2pt}{0.4pt}}}\ULon}
\makeatother

\SetKwInput{KwData}{Input}
\SetKwInput{KwResult}{Output}
\SetKw{Continue}{continue}
\SetKw{IfNot}{not}

\SetCommentSty{mycommfont}

\lstdefinelanguage{sparql}{
    basicstyle=\tiny,
    keywordstyle=\color{eclipseKeywords},
    commentstyle=\color{blue}, % style of comment
    stringstyle=\color{darkgreen},
    keywords={FILTER, NOT, EXISTS, bound, OPTIONAL, BIND, IF, isIRI, AS, INSERT, UPDATE, DELETE, WHERE, SELECT, UNION, rr,rml,ql,rdfs},
    morecomment=[n]{?}{\ },
    morecomment=[l]\#,
    morestring=[b]",
    morestring=[s]{<}{>},
    frame=lines,
}
\lstdefinelanguage{rdf}{
    basicstyle=\scriptsize,
    keywordstyle=\color{eclipseKeywords},
    commentstyle=\color{gray}, % style of comment
    stringstyle=\color{darkgreen},
    keywords={rml, ql, rr, rdfs},
    morecomment=[n]{?}{\ },
    morecomment=[l]\#,
    morestring=[b]",
    morestring=[s]{<}{>},
    frame=lines,
}

\lstdefinelanguage{triples}{
    basicstyle=\tiny,
    keywordstyle=\color{eclipseKeywords},
    commentstyle=\color{darkgreen}, % style of comment
    keywords={a, rml, ql, rr, rdfs},
    morecomment=[s]{?}{\ },
    frame=lines,
}

\newcommand{\symAttrUniverse}{\mathcal{A}}

\newcommand{\symAllStrings}{\mathcal{S}} % the symbol for the set of all strings
\newcommand{\symAllIRIs}{\mathcal{I}} % the symbol for the set of all IRIs
\newcommand{\symAllLiterals}{\mathcal{L}} % the symbol for the set of all literals
\newcommand{\symAllBNodes}{\mathcal{B}} % the symbol for the set of all blank nodes
\newcommand{\symAllVariables}{\mathcal{V}}

\newcommand{\symIRI}{u}
\newcommand{\symPattern}{P}
\newcommand{\symTP}{tp}
\newcommand{\symRDFGraph}{G}
\newcommand{\symDataObjUni}{\mathcal{D}}
\newcommand{\symDataAccUni}{\mathcal{Q}}

\newcommand{\symSourcerefUni}{\mathcal{R}}
\newcommand{\symAttrSubSet}{A}
\newcommand{\symProjSet}{P}

\newcommand{\bigDataObject}{D}
\newcommand{\symSetDataset}{\mathcal{D}^\texttt{\tiny ds}\!}
\newcommand{\symSetDatasetX}[1]{\mathcal{D}^\texttt{\tiny ds}_{#1}}
\newcommand{\symSetContentA}{\mathcal{D}^\texttt{\tiny c1}}
\newcommand{\symSetContentAX}[1]{\mathcal{D}^\texttt{\tiny c1}_{#1}}
\newcommand{\symSetContentB}{\mathcal{D}^\texttt{\tiny c2}}
\newcommand{\symSetContentBX}[1]{\mathcal{D}^\texttt{\tiny c2}_{#1}}
\newcommand{\symDataAcc}{L}
\newcommand{\symDataAccX}[1]{L_{#1}}

\newcommand{\symAttrQueryMap}{\mathbb{P}}
\newcommand{\symDataSeq}{\bar{O}}

\newcommand{\symJoinAttrPairs}{\mathbb{J}}

\newcommand{\IbaseIRI}{\symIRI_\mathsf{base}}

\newcommand{\attr}{a}
\newcommand{\subjAttr}{\attr_\textrm{s}}
\newcommand{\predAttr}{\attr_\textrm{p}}
\newcommand{\objAttr}{\attr_\textrm{o}}

\newcommand{\error}{\epsilon}
\newcommand{\mappingExpr}{M}
\newcommand{\mappingTuple}{t}
\newcommand{\symMappingInst}{I}
\newcommand{\mappingRel}{(\symAttrSubSet, \symMappingInst)}

\newcommand{\queryExpr}{q}

\newcommand{\dataObject}{d}
\newcommand{\eval}{\mathit{eval}}
\newcommand{\evalX}[1]{\eval_{#1}}
\newcommand{\cbeval}{\mathit{eval}'\!}
\newcommand{\cbevalX}[1]{\eval_{#1}'}

\newcommand{\sourceReference}{sr}
\newcommand{\sourceType}{type}

\newcommand{\sourceTypeTuple}{(\symSetDataset, \symSetContentA\!, \symSetContentB\!, \symDataAcc, \symDataAcc'\!, \eval, \cbeval, \typeCast)}
\newcommand{\sourceTypeTupleX}[1]{\sourceTypeTupleXXX{#1}{#1}{#1}}
\newcommand{\sourceTypeTupleXXX}[3]{(\symSetDatasetX{#1}, \symSetContentAX{#2}, \symSetContentBX{#3}, \symDataAccX{#2}, \symDataAccX{#3}', \evalX{#2}, \cbevalX{#3}, \typeCastX{#1})}

\newcommand{\typeCast}{cast}
\newcommand{\typeCastX}[1]{\typeCast_{#1}}

\newcommand{\var}{v}

\newcommand{\extExpr}{\varphi}

\newcommand{\extend}{\textsf{\small Extend}}
\newcommand{\extOp}{\extend_{\varphi}^{\attr}}

\newcommand{\equiJoinOp}{\textsf{\small EqJoin}^{\symJoinAttrPairs}}

\newcommand{\toIRI}{\texttt{toIRI}}
\newcommand{\toBNode}{\texttt{toBNode}^{\LTB}}
\newcommand{\toLiteral}{\texttt{toLiteral}}

\newcommand{\concat}{\texttt{concat}}

\newcommand{\lex}{\mathit{lex}}
\newcommand{\dt}{\mathit{dt}}
\newcommand{\literalTuple}{(\lex, \dt)}
\newcommand{\LTB}{\mathit{S2B}}

\newcommand{\projectOp}[1]{\textsf{\small Project}^{#1}}

\newcommand{\union}{\textsf{\small Union}}
\newcommand{\extract}{\textsf{\small Extract}}
\newcommand{\validInput}{\sigma}

\newcommand{\symAllTORBs}{\Phi_\mathsf{torb}}

\newcommand{\fctDom}[1]{\mathrm{dom}(#1)}

\newcommand{\fctAttrs}[1]{\mathrm{attrs}(#1)}

\newcommand{\fctEval}[2]{\eval(#1,#2)}

\newcommand{\fctCbeval}[3]{\cbeval(#1,#2,#3)}

\newcommand{\fctTypeCast}[1]{\typeCast(#1)}

\newcommand{\fctExtOp}[1]{\extOp(#1)}

\newcommand{\fctExtOpX}[3]{\textsf{\small Extend}^{#1}_{#2}(#3)}

\newcommand{\fctEquiJoinOp}[2]{\equiJoinOp(#1, #2)}

\newcommand{\fctToIRI}[2]{\toIRI(#1,#2)} % arguments: term, base IRI
\newcommand{\fctToBNode}[1]{\toBNode(#1)}
\newcommand{\fctToLiteral}[1]{\toLiteral(#1)}

\newcommand{\fctConcat}[2]{\concat(#1, #2)}

\newcommand{\fctProjectOpDflt}[1]{\fctProjectOp{\symProjSet\!}{#1}}
\newcommand{\fctProjectOp}[2]{\projectOp{#1}(#2)}

\newcommand{\fctUnion}[2]{\union(#1,#2)}
\newcommand{\fctUnionBig}[2]{\union\bigl(#1,#2\bigr)}
\newcommand{\fctExtractDflt}{\extract_{\sourceReference}^{(\sourceType, \queryExpr, \symAttrQueryMap)}}
\newcommand{\fctExtractX}[4]{\extract_{#1}^{(#2, #3, #4)}}

\newcommand{\fctValidInput}[1]{\validInput(#1)}
\newcommand{\fctApply}[2]{#1[#2]}
\newcommand{\evalP}[2]{[\![#1]\!]_{#2}}
\newcommand{\fctTrmaps}[1]{\textrm{TrMaps}(#1)}

\newcommand{\OpAND}{\,\textsf{\small AND}\,}
\newcommand{\OpOPT}{\,\textsf{\small OPT}\,}

\newcommand{\ttl}[1]{\texttt{\small #1}}  % for conrete IRIs, etc. used in examples and in formulas
\hypersetup{bookmarksopen=true, citecolor=blue}
\begin{document}												% end of preamble and beginning of text that will be printed

%% \title{\texorpdfstring{On the Satisfiability of\\ RML with SPARQL}{On the Satisfiability of RML with SPARQL}}
%% \titlerunning{On the Satisfiability of RML with SPARQL}
\title{Query-Specific Pruning of RML Mappings}
%\title{Satisfiability-Based Pruning of RML Mappings for SPARQL Querying}

\ExtendedVersion{
\subtitle{
(Extended Version)%
\thanks{This is an extended version of a research paper accepted for the 25th
	%International
	Int.\
Semantic Web Conference (ISWC~2026). The extension consists of
	an appendix of proofs.}
	%a few additional equivalences in Section~\ref{ssec:equivalences:ProjectionPushing} and the complete content of Section~\ref{ssec:equivalences:ExtendPushingOrPulling}, which are not included in the version prepared for the conference proceedings.}
}
}

\author{%
 	Sitt Min Oo\inst{1}\orcidlink{0000-0001-9157-7507}
 	\and
 	Olaf Hartig\inst{2}\orcidlink{0000-0002-1741-2090}%
%Anonymous Authors
}
 	\authorrunning{S. Min Oo and O. Hartig}
%	\authorrunning{Anonymous Authors}
 \institute{
 	Ghent University - imec, Ghent, Belgium \\
 	\email{x.sittminoo@ugent.be}
 	\and
 	Linköping University, Linköping, Sweden \\
 	\email{olaf.hartig@liu.se}
 }

\maketitle              % typeset the header of the contribution

\ExtendedVersion{
	\vspace{-8mm} % Layout Adjustment
}

\begin{abstract}
Current approaches for knowledge graph construction with RML focus on full RDF graph materialization without considering user queries%
	%. These approaches are
	, which is
inefficient in dynamic query environments where often only a specific subset of the full graph is needed to answer a given query. This paper introduces an approach to prune RML mappings such that
	%a given SPARQL query can still be answered completely over the partially-materialized graph.
	the resulting partially-materialized graph is still sufficient to answer a given SPARQL query completely.
By evaluating the approach based on
	a well-know RML materialization
	%the GTFS-Madrid
benchmark, we show that such pruning significantly reduces both the materialization time and the size of the produced graph, while also noticeably reducing querying time.
\end{abstract}

\ExtendedVersion{
	\vspace{-8mm} % Layout Adjustment
	\enlargethispage{\baselineskip}%  % Layout Adjustment
}

% \enlargethispage{\baselineskip}%  % Layout Adjustment

\section{Introduction}
\label{sec:introduction}

State-of-the-art knowledge graph generation approaches support
	%declarative
mapping languages such as R2RML~\cite{Seema2012R2RML},
	RML~\cite{dimou_ldow_2014,IglesiasMolina2023:RMLOntology},
and SPARQL-Generate~\cite{lefrancois_eswc_2017} to provide access to RDF views of other forms of structured data
	%such as relational databases,
	(e.g., relational databases)
as well as semi-structured data
	%such as JSON and XML.
	(e.g., JSON, XML).
These approaches either materialize the RDF view~\cite{Iglesias2022ScalingKnowledgeGraph,ArenasGuerrero2022MorphKGCScalable,Freund2024FlexRML,MinOo2022RMLStreamer,minoo2024rmlweaver}, making it available for querying or further processing directly as RDF graphs, or translate SPARQL queries over the RDF views into a query language supported by the underlying data sources~\cite{ChavesFraga2021Enhancingvirtualontology,calvanese2017ontop,Lopes2011XSPARQL}.
	%can be broadly categorized into two classes: i)~materialization approaches actually generate the RDF representation, making it available for querying or further processing directly as RDF graphs, whereas ii)~virtualization approaches translate SPARQL queries over the RDF views into a query language supported by the underlying data sources.

In the context of use cases that require integrated query access over a federation of multiple data sources, including non-RDF ones, we observe that both of these two types of approaches pose practical limitations: Relying on query translation~(also called virtualization) would limit a federation engine to types of data sources that support a query language. While materialization approaches, in contrast, can access arbitrary types of data sources, including those without a query language, materialization approaches can easily become
	%impractically
inefficient in this setting as they are designed to always produce the full RDF view of the source data, even if the part(s) of the federation query to be answered via this view need only a smaller part of the view to be answered completely.

% Considering these two types of approaches in the context of use cases that require integrated query access over a federation of multiple data sources, including non-RDF data sources, we observe that both of them pose practical limitations:
% Relying on query translation~(also called virtualization) would limit a
% federation engine to types of data sources that support a query language. 
% In contrast, materialization approaches can, in principle, access arbitrary types of data sources, including those without a query language. 
% Recent work by Hartig and Westman~\cite{Hartig26:HeFQUINWebAPIsDemo} leverages
% this flexibility to extend a query federation engine with support for
% JSON-based Web APIs {\changed that update data frequently}, materializing RDF views of the API responses at runtime. 
% Yet, while enabling such an engine to support more diverse types of federation members, materialization approaches can easily become
% 	%impractically
% inefficient in this setting as they are designed to always produce the full RDF view of the source data, even if
% 	the part(s) of the federation query to be answered via this view need only a smaller part of the view to be answered completely.

This work introduces an approach that combines the support for a wider range of data sources from materialization with the query-aware efficiency of virtualization. The approach focuses on mappings described in RML and the core idea is to prune away the parts of an RML mapping that define portions of the resulting RDF view that are irrelevant for answering a given SPARQL query%
	~(which may be a sub-query assigned to a non-RDF data source within a federation query~\cite{Hartig26:HeFQUINWebAPIsDemo}). The
	%; the
pruned mapping can then be used to materialize a smaller RDF graph for which the query returns the same result as for the full RDF view. Section~\ref{sec:demonstration_pruning} illustrates the approach with an example.

% paper outline
Our main technical contribution is an algorithm that captures the pruning approach formally~(Section~\ref{sec:pruning_algorithm}), for which we assume that the mapping to be pruned is given in the algebraic form introduced in our previous work~\cite{minoo2025AlgebraPublished}~(%
	%and
summarized in Section~\ref{sec:Preliminaries}).
As a first
	%step of our work now,
	step,
we define a notion of satisfiability of SPARQL graph patterns over such algebraic mapping expressions and show that this notion of satisfiability is undecidable, which constitutes another technical contribution~(Section~\ref{sec:Satisfiability}).
We then introduce a syntactic property to identify a class of cases in which a triple pattern is guaranteed to be not satisfiable over a specific mapping expression, and we show the correctness of this property~(Section~\ref{sec:Incompatibility}). This
	formal
result provides the foundation of the~pruning~algorithm.

As our second main contribution, we evaluate the effectiveness of our pruning approach based on the GTFS-Madrid benchmark~(Section~\ref{sec:evaluation}); our evaluation
	confirms
	%shows
that pruning can reduce
	%both the materialization time and the size of the resulting RDF graph significantly, while the pruning time is negligible.
the materialization time significantly (down to at most 8\% of the full materialization time in 2/3 of the considered cases), while the pruning time is negligible. Moreover, the resulting RDF graphs are much smaller as well, which can also lead to a noticeable reduction of query times.

\enlargethispage{\baselineskip}%  % Layout Adjustment

\section{Demonstration of the Approach}%
\label{sec:demonstration_pruning}%
\begin{wrapfigure}[11]{r}{0.65\textwidth}%
\vspace{-9mm}
    \begin{lstlisting}[caption=Snippet of an RML mapping to convert airport data into an RDF graph., label=lst:rml_input,frame=lrtb, basicstyle=\scriptsize\ttfamily, numbers=right, captionpos=b]
...
rml:subjectMap [ rml:reference "airport_id" ;
                 rml:termType rml:IRI ];
rml:predicateObjectMap [
 rml:predicate ex:route;
 rml:objectMap [ rml:template 
    "http://example.com/route/{transitRoute}" ] ];
rml:predicateObjectMap [
 rml:predicate gtfs:long;
 rml:objectMap [ rml:reference "longitude";
                 rml:datatType xsd:double ] ].
    \end{lstlisting}%
\end{wrapfigure}%
\noindent
We begin by illustrating
	%the intuition of
our pruning approach informally, for which we assume familiarity with the concepts of
	%RML~\cite{dimou_ldow_2014,IglesiasMolina2023:RMLOntology} and SPARQL.
	RDF, SPARQL, and RML.
Listing~\ref{lst:rml_input} presents a snippet of an example RML document that defines a mapping to generate RDF triples 
about the transit routes (lines~4--7) and the longitude coordinate (lines~8--11) of airports.
Listing~\ref{lst:sparql_input} provides a SPARQL query with two triple patterns to query the resulting RDF data.
We go through each predicate-object map of the RML mapping and decide if it can be pruned by considering each of the triple patterns of the
	%SPARQL
query.

\begin{wrapfigure}[6]{r}{0.67\textwidth}
\vspace{-9mm}
    \begin{lstlisting}[caption=A SPARQL query to retrieve information about airports (prefix declarations omitted)., label=lst:sparql_input,frame=lrtb, basicstyle=\scriptsize\ttfamily, captionpos=b]
SELECT * WHERE {
 ?airportId ex:route <http://transit.api/route/43> .
 ?airportId gtfs:long "23.0"^^xsd:double .  }
    \end{lstlisting}
\end{wrapfigure}
	%To check the first predicate-object map (lines~4--7),
	%To check the first predicate-object map (Listing~\ref{lst:rml_input} lines~4--7),
	For the first pre\-di\-cate-object map (lines 4--7),
we first consider the first triple pattern. % (Listing~\ref{lst:sparql_input}). 
While the subject and the predicate of possible triples produced from the predicate-object map match the triple pattern, 
the object of the triple pattern  cannot be produced by the corresponding object-term map. 
To determine this
	type of
unsatisfiability, we transform the \ttl{rml:template} value of the object-term map,~%
"\ttl{http://example.com/route/\{transitRoute\}}", into the regular expression "\ttl{http:\textbackslash/\textbackslash/example.com\textbackslash/route\textbackslash/.+}".
It is evident that the object IRI of the triple pattern, which begins with the substring \ttl{"http://transit.api"}, does not match the regular expression. 
Thus, the first triple pattern is not satisfiable for the first predicate-object map.
Moving to the second triple pattern, it is also evident that the first predicate-object map cannot produce triples that match the pattern.
The predicate-term is the constant \ttl{ex:route}, which is not the same as the predicate IRI \ttl{gtfs:long} of the triple pattern.
We thus conclude that, since none of the triple patterns is satisfiable with it, the first predicate-object map~can~be~pruned.

For the second predicate-object map (%
  %Listing~\ref{lst:rml_input}
lines~8--11), triples produced from it do not match the first triple pattern because they would have \ttl{gtfs:long} as predicate,
	%instead of
	not
\ttl{ex:route}.
However, the second triple pattern may be satisfiable with that predicate-object map since the predicate IRIs match in this case and, for checking the object terms, we apply the following procedure. 
We first transform the \ttl{rml:reference} value \ttl{"long"} of the object-term map into the regular expression \ttl{".+"}. 
The lexical form of the object literal in the second triple pattern matches this regular expression. 
Moreover, the datatype declared in the object-term map, \ttl{xsd:double}, matches the
	%expected
datatype of the object literal in the triple pattern. 
Thus, at least one triple pattern may be satisfiable with the second predicate-object map, which means this predicate-object map should not be pruned away.
	Hence, at the end, only the second predicate-object map of the snippet of RML in Listing~\ref{lst:rml_input} is kept after pruning. 

The \emph{main use case of this approach} is SPARQL-based query federation with non-RDF data sources that cannot be accessed via a query language. Since
	%(query-rewriting-based)
virtualization is not an option for such sources, their data needs to be converted directly to RDF to be queried via SPARQL. As a typical example of such cases, we refer to recent work by Hartig and Westman~\cite{Hartig26:HeFQUINWebAPIsDemo} which integrates REST/Web APIs within a federation engine that materializes RML-based RDF representations of the API data at query time. Query-time materialization is relevant in this setting because each federation query may require data from different API endpoints; moreover, the API data may change often (e.g., a weather API~\cite{Hartig26:HeFQUINWebAPIsDemo}).

\ExtendedVersion{
	\enlargethispage{\baselineskip}%  % Layout Adjustment
}

% \vspace{-2mm} % Layout Adjustment
\section{Preliminaries}
\label{sec:Preliminaries}
% \vspace{-2mm} % Layout Adjustment

	%To define our pruning approach formally, we adopt
	We formalize our pruning approach in terms of
	the RML-related mapping algebra of our earlier work%
	%the RML-related mapping algebra of Min Oo and Hartig%
~\cite{minoo2025AlgebraPublished}.
	To provide the relevant background for this formalization, this section introduces the concepts of that algebra, and of RDF~and~SPARQL.
	%This section provides the relevant background for this formalization by introducing the concepts of that algebra, and of RDF~and~SPARQL.
	%This section provides the relevant background for this work by introducing the concepts of that algebra, and of RDF~and~SPARQL.
	%This section introduces the relevant concepts of that algebra, and of RDF and SPARQL.

% \vspace{-2mm} % Layout Adjustment
\subsection{Relevant Concepts of RDF and SPARQL}
\label{ssec:Preliminaries:RDFandSPARQL}
% \vspace{-2mm} % Layout Adjustment

Let $\symAllStrings$ be the countably infinite set of all possible strings, and $\symAllIRIs$ 
be the subset of $\symAllStrings$ that consists of all IRIs.
$\symAllLiterals$ is the countably infinite set of all RDF literals where every such literal is a pair $\literalTuple \in \symAllStrings \times \symAllIRIs$ in which $\lex$ is the lexical form and $\dt$ is the datatype IRI of the literal.
Moreover, $\symAllBNodes$ is the countably infinite set of blank nodes (and is disjoint from $\symAllStrings$ and $\symAllLiterals$).
IRIs, literals, and blank nodes are jointly referred to as \emph{RDF terms}.
An \emph{RDF triple} is a tuple $(s,p,o) \in (\symAllIRIs\cup\symAllBNodes) \times \symAllIRIs \times (\symAllIRIs\cup\symAllBNodes\cup\symAllLiterals)$, and an \emph{RDF graph} is a set of RDF triples.

	For our definitions related to SPARQL we adopt the algebraic syntax of P{\'e}rez et al.~\cite{PerezAG09}. Queries
	%Regarding SPARQL, queries
	%SPARQL queries
are formed using \emph{graph patterns}, of which the most basic type is a \emph{triple pattern}, that is, a tuple $(s,p,o) \in (\symAllIRIs\cup\symAllVariables) \times (\symAllIRIs\cup\symAllVariables) \times (\symAllIRIs\cup\symAllLiterals\cup\symAllVariables)$, where $\symAllVariables$ is a countably infinite set of variables (disjoint from $\symAllStrings$, $\symAllLiterals$, and $\symAllBNodes$). Other graph patterns can then be constructed recursively, using operators such as $\OpAND$ and $\OpOPT\!$~\cite{PerezAG09}.
The result of evaluating any such graph pattern~$\symPattern$ over an RDF graph~$\symRDFGraph$ is a set, denoted by $\evalP{\symPattern}{\symRDFGraph}$, that consists of so-called \emph{solution mappings}, which are partial functions of the form $\mu \!: \symAllVariables \rightarrow \symAllIRIs\cup\symAllBNodes\cup\symAllLiterals$. If $\symPattern$ is a triple pattern $\symTP=(s,p,o)$, then $\evalP{\symTP}{\symRDFGraph}$ consists of every solution mapping~$\mu$ for which $\fctDom{\mu} = \{s,p,o\} \cap \symAllVariables$ and $\mu[\symTP] \in \symRDFGraph$, where $\mu[\symTP]$ denotes the triple obtained by replacing the variables in $\symTP$ according to $\mu$. For other forms of graph patterns, we refer to P{\'e}rez et al.'s work for the definition of $\evalP{\symPattern}{\symRDFGraph}$~\cite{PerezAG09}.

% \vspace{-2mm} % Layout Adjustment

\subsection{Data Model of the Mapping Algebra}
\label{ssec:Preliminaries:DataModel}
% \vspace{-2mm} % Layout Adjustment

The mapping algebra
	of our earlier work
	%by Min Oo and Hartig
is defined over so-called \emph{mapping relations}~\cite{minoo2025AlgebraPublished} in which the possible values are RDF terms, plus a special value, $\error$, that
	%is used to capture
	captures
processing errors (and is not an RDF term). For the
	following
formal definition of these relations,
	%and of the tuples therein,
let $\symAttrUniverse$ be a countably infinite set of \emph{attributes}.

\begin{definition}
	\normalfont
	\hspace{-3mm}$^\text{\cite{minoo2025AlgebraPublished}}$
	A \textbf{mapping tuple} is a partial function $\mappingTuple \!: \symAttrUniverse \rightarrow \symAllIRIs\cup\symAllBNodes\cup\symAllLiterals \cup \{\error\}$
\end{definition}

\begin{definition}
	\label{def:MappingRelation}
	\normalfont
	\hspace{-3mm}$^\text{\cite{minoo2025AlgebraPublished}}$
	A \textbf{mapping relation}~$r$ is a tuple $\mappingRel$, where
	$\symAttrSubSet \subset \symAttrUniverse$ is a finite, non-empty set of attributes
	and
	$\symMappingInst$ is a set of mapping tuples such
	that, for every such tuple~$\mappingTuple \in \symMappingInst$, it holds that $\mathrm{dom}(\mappingTuple) = \symAttrSubSet$.
\end{definition}

While the mapping algebra operates over such mapping relations, the final mapping relation
	that results from a relevant sequence of such operations
is meant to capture an RDF dataset.
	At this point we diverge slightly from the original formalism~\cite{minoo2025AlgebraPublished}, which considers the creation of whole RDF datasets~(i.e., including named graphs);
	%While Min Oo and Hartig~\cite{minoo2025AlgebraPublished} consider the creation of whole RDF datasets~(i.e., including named graphs),
in this paper we limit ourselves to single RDF graphs. 
For this purpose, we assume three special attributes, $\subjAttr, \predAttr, \objAttr \in \symAttrUniverse$, and adapt the definition of an RDF representation of mapping relations as follows.

\begin{definition}
	\label{def:ResultingRDFGraph}
	\normalfont
	Let $r = \mappingRel$ be a mapping relation with $\subjAttr, \predAttr, \objAttr \in \symAttrSubSet$.
	The \textbf{RDF graph resulting from $r$} is the RDF graph
	\begin{equation*}
		\symRDFGraph =
		\bigl\{
		\bigl( \mappingTuple(\subjAttr), \mappingTuple(\predAttr), \mappingTuple(\objAttr) \bigr)
		\,\big|\,
		\mappingTuple \in \symMappingInst
		\text{ such that }
		\bigl( \mappingTuple(\subjAttr), \mappingTuple(\predAttr), \mappingTuple(\objAttr) \bigr) \text{ is an RDF triple}
		\bigr\}
		.
	\end{equation*}
\end{definition}

For examples of
	%mapping relations and their RDF representations,
	these concepts,
refer to
	our earlier
	%Min Oo and Hartig's
work~\cite{minoo2025AlgebraPublished}.

\subsection{Syntax of RML-Specific Mapping Expressions}
\label{ssec:Preliminaries:RMLSpecMappingExpressions:Syntax}

% \vspace{-2mm} % Layout Adjustment

%This section introduces the syntax of the RML-spe\-cif\-ic fragment of Min Oo and Hartig's mapping algebra, which is the fragment that we focus on in this paper.  ...

This section introduces the syntax of the mapping expressions that we consider in this paper, which is based on
	our earlier-introduced
	%Min Oo and Hartig's
mapping algebra~\cite{minoo2025AlgebraPublished}.

While this algebra is of a more general nature, not specific to any concrete mapping language such as RML,
	we also introduced a translation of RML into the
	%the authors also introduce a translation of RML into their
algebra~\cite[Section~5]{minoo2025AlgebraPublished}. The fragment of the algebra that this translation uses is the focus of our work in this paper. Therefore, instead of re-introducing the complete algebra here, we introduce only the relevant types of expressions that
	%can be captured in the RML-spe\-cif\-ic fragment that covers
	cover
any possible output of translating RML into the algebra as per the translation algorithm of
	our earlier work%
	%Min Oo and Hartig%
~\cite{minoo2025AlgebraPublished}.
To define these types of mapping expressions we need to introduce a number of related concepts first.

We begin with concepts related to the $\extract$ operator\footnote{In the original work this operator is called \textsf{\small Source}~\cite{minoo2025AlgebraPublished}. We have renamed it for this paper because the name $\extract$ captures more clearly the purpose of this operator.}
of the algebra, which initializes a mapping relation
	that provides a relational view of data that can be
	%based on data
extracted from input data sources. The definition of this operator is based on an abstraction of source data and corresponding query languages: The infinite sets~$\symDataObjUni$ and~$\symDataAccUni$ capture all possible data objects and all possible query languages, respectively. Examples of data objects are: the whole content of a particular JSON file, a single JSON object, and a value of a JSON field. Data objects of the same kind~(e.g., all possible JSON objects) would be captured as a dedicated subset of~$\symDataObjUni$. Each query language $\symDataAcc \in \symDataAccUni$ is considered as a set, where every element $\queryExpr \in \symDataAcc$ is one of the queries written in
	%that language.
	$\symDataAcc$.
Types of data sources, as considered by the $\extract$ operator, are captured as a tuple $\sourceType = \sourceTypeTuple$ where
$\symSetDataset \subseteq \symDataObjUni$ specifies the kind of data objects that can be accessed from data sources of this type;
$\symDataAcc \in \symDataAccUni$ is a language to enumerate components of any data object in $\symSetDataset$, where $\symSetContentA\! \subseteq \symDataObjUni$ is the set of all possible such components;
$\symDataAcc' \in \symDataAccUni$ is a language to select values from the components, with $\symSetContentB\! \subseteq \symDataObjUni$ being the set of all values possible for the type of data source;
$\typeCast$ is a function defining how these values map to RDF literals;
and $\eval$ and $\cbeval$ are functions that define the evaluation semantics of $\symDataAcc$ and $\symDataAcc'\!$, respectively. Further details and examples are in the original paper~\cite{minoo2025AlgebraPublished}.

RML mappings contain references to the data sources
	from which the input data is meant to be obtained
	%to which they are meant to be applied
(e.g., file names).
	As a corresponding abstraction, we assume
	%The corresponding abstraction is
a countably infinite set~$\symSourcerefUni$ of so-called \emph{source references}. As we shall see, every $\extract$ operator
	%in a mapping expression
is parameterized with such a source~reference.

Another relevant operator is \textsf{\small Extend}, which is parameterized with an attribute~%
	%$\attr \in \symAttrUniverse$
	$\attr$
and a so-called \emph{extend expression}~\cite{minoo2025AlgebraPublished} that can be evaluated with a mapping tuple as input and that produces an RDF term or the error symbol~$\error$ as output. For every input tuple,
	%the \textsf{\small Extend} operator
	\textsf{\small Extend}
uses this output to extend the tuple with a value for attribute~$\attr$.
	%A specific type of extend expressions that we introduce specifically for the considered fragment of the mapping algebra is defined~as~follows.
	%The following definition captures the specific type of extend expressions used by the fragment of the mapping algebra considered in this paper.
	The specific type of extend expressions used by the fragment of the mapping algebra considered in this paper is defined~as~follows.

\begin{definition}
	\label{def:TORBExtendExpression}
	\normalfont
	A \textbf{template-or-reference-based extend expression (torb-extend expression)} is an extend expression%
		%~\cite{minoo2025AlgebraPublished}%
	~$\extExpr$ of any of the following forms:
	\begin{enumerate}
		\item $\extExpr$ is an RDF literal~$(lex, dt)$ with $dt=\ttl{xsd:string}$.

		\item $\extExpr$ is an attribute in $\symAttrUniverse$.

		\item $\extExpr$ is of the form $\fctConcat{\extExpr_1, \ldots}{\extExpr_n}$, where $n \geq 2$ and every $\extExpr_i$ ($1 \leq i \leq n$) is either a literal~$(lex, dt)$ with $dt=\ttl{xsd:string}$ or an attribute in $\symAttrUniverse$.
	\end{enumerate}
		Hereafter, we
		%We
	write $\symAllTORBs$ to denote the set of all possible torb-extend expressions and, for every such expression~$\extExpr \in \symAllTORBs$, $\fctAttrs{\extExpr}$ is the set of all attributes in $\extExpr$.
\end{definition}

% \begin{example}
% 	\label{ex:TORBExtendExpression}
% 	\inlinetodo{small example that uses an extend expression of an earlier example}
% \end{example}

The forms of
	%RML-spe\-cif\-ic mapping
	algebra
expressions that
	the RML-to-algebra translation algorithm~\cite{minoo2025AlgebraPublished} produces and, thus, that
we consider in this paper
	%are all built on
	all share
a common form of sub-expressions, which result from translating individual RML triples maps into the
	%mapping
algebra. As the last ingredient needed for defining the
	considered fragment of the algebra,
	%RML-spe\-cif\-ic algebra expressions,
	%RML-spe\-cif\-ic mapping expressions,
%
	%we introduce this form of sub-expressions in the following definition.
	we introduce this form of sub-expressions:

\begin{definition}
	\label{def:TrMapExpression}
	\normalfont
	Let $\IbaseIRI$ be an IRI (considered as \emph{base IRI}) and $\LTB \!: \symAllStrings \rightarrow \symAllBNodes$ be an injective function that maps every string to a unique blank node.
	A \textbf{Triples-Map-spe\-cif\-ic expression (TrMap-ex\-pres\-sion)} with $\IbaseIRI$ and $\LTB$ is either%
% 	\begin{equation*}
% % 		\fctExtOpX{\graphAttr}{\extExpr_\mathrm{g}}{
% 			\fctExtOpX{\objAttr}{\extExpr_\mathrm{o}}{ \mappingExpr }
% % 		}
% 		%
% 		\quad\text{or}\quad
% 		%
% % 		\fctExtOpX{\graphAttr}{\extExpr_\mathrm{g}}{
% 			\fctExtOpX{\objAttr}{\extExpr_\mathrm{o}'}{
% 				\fctEquiJoinOp{ \mappingExpr }{
% 					\fctExtractX{\sourceReference'}{\sourceType'\!}{\queryExpr'\!}{\symAttrQueryMap'}
% 				}
% 			}
% % 		}
% 	\end{equation*}
	\begin{align*}
		&
		\fctExtOpX{\objAttr}{\extExpr_\mathrm{o}}{
			\;
			\hspace{12mm}
			\fctExtOpX{\predAttr}{\extExpr_\mathrm{p}}{
				\fctExtOpX{\subjAttr}{\extExpr_\mathrm{s}}{
					\fctExtractDflt
				}
			}
			\;
		}
		\quad \text{ or}
		\\
		&\fctExtOpX{\objAttr}{\extExpr'_\mathrm{o}}{
			\;
			\fctEquiJoinOp{
				\fctExtOpX{\predAttr}{\extExpr_\mathrm{p}}{
					\fctExtOpX{\subjAttr}{\extExpr_\mathrm{s}}{
						\fctExtractDflt
					}
				}
			}{
				\fctExtractX{\sourceReference'}{\sourceType'\!}{\queryExpr'\!}{\symAttrQueryMap'}
			}
			\;
		}, 
	\end{align*}
	where:
	\begin{enumerate}
% 		\item $\mappingExpr$ is a sub-expression of the form
% 		$\fctExtOpX{\predAttr}{\extExpr_\mathrm{p}}{
% 			\fctExtOpX{\subjAttr}{\extExpr_\mathrm{s}}{
% 				\fctExtractDflt
% 			}
% 		}$;
% 
% 		\item $\extExpr_\mathrm{g}$ is an extend expression that is either
% 		\begin{enumerate}
% 			\item the IRI \ttl{rml:defaultGraph} or
% 			\item of the form $\fctToIRI{\extExpr'}{\IbaseIRI}$ with $\extExpr'\! \in \AllUntypedExprs{\fctDom{\symAttrQueryMap}}$;
% 		\end{enumerate}
%
		\item $\extExpr_\mathrm{o}$ is an extend expression~\cite{minoo2025AlgebraPublished} of any of the following specific forms:
		\begin{enumerate}
			\item[(a)]
			a literal,
			\quad
			(b) an IRI,
			\quad
			(c) a blank node,

			\item[(d)]
			$\toLiteral(\extExpr'\!, dt)$ with $\extExpr'\! \in \symAllTORBs$, $\fctAttrs{\extExpr} \subseteq \fctDom{\symAttrQueryMap}$,
				and
			$dt \in \symAllIRIs$,

			\item[(e)]
			$\fctToIRI{\extExpr'}{\IbaseIRI}$ with $\extExpr'\! \in \symAllTORBs$ and $\fctAttrs{\extExpr} \subseteq \fctDom{\symAttrQueryMap}$, or

			\item[(f)]
			$\fctToBNode{\extExpr'}$ with $\extExpr'\! \in \symAllTORBs$ and $\fctAttrs{\extExpr} \subseteq \fctDom{\symAttrQueryMap}$;
		\end{enumerate}

		%%% Note that the restrictions captured by the following two bullet points
		%%% assume that the RML mapping that has been translated into a mapping
		%%% expression is valid with respect to the RML-Core spec, which requires
		%%% that, e.g., a predicate map "is a rule that MUST generate an IRI"
		%%%                     -Olaf
		\item $\extExpr_\mathrm{p}$ is an extend expression that may be only of the form~(b) or~(e);

		\item $\extExpr_\mathrm{s}$ and $\extExpr'_\mathrm{o}$ are extend expressions
			%that may be only
		of the form~(b), (c), (e), or (f), respectively;

		\item $\sourceReference$ and $\sourceReference'\!$ are source references (potentially the same);

		\item
		$\sourceType = \sourceTypeTupleX{1}$ is a source type,
		$\queryExpr \in \symDataAccX{1}$, and
		$\symAttrQueryMap \!: \symAttrUniverse \rightarrow \symDataAccX{1}'$ is a partial function such that $\fctDom{\symAttrQueryMap} \cap \{\subjAttr, \predAttr, \objAttr%
			%, \graphAttr
		\} = \emptyset$;

		\item $\sourceType' = \sourceTypeTupleX{2}$ is a source type, 
		$\queryExpr'\! \in \symDataAccX{2}$, and
		$\symAttrQueryMap' \!: \symAttrUniverse \rightarrow \symDataAccX{2}'$ is a partial function such that $\fctDom{\symAttrQueryMap'} \cap \{\subjAttr, \predAttr, \objAttr%
			%, \graphAttr
		\} = \emptyset$;

		\item $\fctDom{\symAttrQueryMap} \cap \fctDom{\symAttrQueryMap'} = \emptyset$;

		\item $\symJoinAttrPairs \subseteq \fctDom{\symAttrQueryMap} \times \fctDom{\symAttrQueryMap'}$.
	\end{enumerate}
\end{definition}

\begin{example}
	\label{ex:TrMapExpression}
	The triples map that consists of the first pre\-di\-cate-object map of the RML mapping in Listing~\ref{lst:rml_input} is captured by the following TrMap-ex\-pres\-sion:
	\begin{equation*}
		\fctExtOpX{\objAttr}{ \fctToIRI{ \fctConcat{\ell}{\attr_1} }{\IbaseIRI} }{
			\fctExtOpX{\predAttr}{ \symIRI_1 }{
				\fctExtOpX{\subjAttr}{ \fctToIRI{\attr_2}{\IbaseIRI} }{
					\fctExtractX{\sourceReference}{\sourceType}{\queryExpr}{\symAttrQueryMap_1}
				}
			}
		},
	\end{equation*}
	where
		$\attr_1$ and $\attr_2$ are attributes different from $\subjAttr$, $\predAttr$, and $\objAttr$, respectively,
		%$\attr_1, \attr_2 \in \symAttrUniverse$,
	$\ell$ is the literal $(\ttl{"http://example.com/route/"}, \ttl{xsd:string})$,
	$\symIRI_1$ is the IRI \ttl{ex:route},
	$\symAttrQueryMap_1 = \{ \attr_1 \mapsto \ttl{"transitRoute"}, \attr_2 \mapsto \ttl{"airport\_id"} \}$,
	and $\sourceReference$ is an arbitrary source reference.
	The other two arguments of the $\extract$ operator, $\sourceType$ and $\queryExpr$, are not specified further
		in this example
	as they depend on the \ttl{rml:logicalSource}
		of the triples map, which is
		%and its \ttl{rml:iterator}, which are
	not
		%given
	in Listing~\ref{lst:rml_input} (and is irrelevant for our pruning~approach).
% 
% 	The TrMap-ex\-pres\-sion for the triples map with the second pre\-di\-cate-object map in Listing~\ref{lst:rml_input} is:
% 	\begin{equation*}
% 		\fctExtOpX{\objAttr}{ \toLiteral(\attr_1, \ttl{xsd:double}) }{
% 			\fctExtOpX{\predAttr}{ \symIRI_2 }{
% 				\fctExtOpX{\subjAttr}{ \fctToIRI{\attr_2}{\IbaseIRI} }{
% 					\fctExtractX{\sourceReference}{\sourceType}{\queryExpr}{\symAttrQueryMap_2}
% 				}
% 			}
% 		},
% 	\end{equation*}
% 	where $\attr_1, \attr_2 \in \symAttrUniverse \setminus \{ \subjAttr, \predAttr, \objAttr \}$,
% 	$\symIRI_2$ is the IRI \ttl{gtfs:long},
% 	$\symAttrQueryMap_2 = \{ \attr_1 \mapsto \ttl{"longitude"},$ $\attr_2 \mapsto \ttl{"airport\_id"} \}$,
% 	and $\sourceReference$, $\sourceType$, and $\queryExpr$ are the exact same as above.
\end{example}

We
	%are now ready to
	can now
define the notion of an RML-spe\-cif\-ic mapping expression. Informally, such an expression is built by wrapping a TrMap-ex\-pres\-sion into a \textsf{\small Project} operator that keeps only the attributes $\subjAttr$, $\predAttr$, and $\objAttr$, and by combining multiple such \textsf{\small Project}-wrapped TrMap-ex\-pres\-sions using \textsf{\small Union} operators.~Formally:

\begin{definition}
	\label{def:RMLSpecificExpression}
	\normalfont
	Let  $\IbaseIRI$ be an IRI
		%(considered as \emph{base IRI})
	and $\LTB \!: \symAllStrings \rightarrow \symAllBNodes$ be an injective function that maps every string to a unique blank node.
	An \textbf{RML-spe\-cif\-ic mapping expression} with $\IbaseIRI$ and $\LTB$ is defined recursively as follows:
	\begin{enumerate}
		\item For every TrMap-ex\-pres\-sion $\mappingExpr_\mathsf{TM}$ with $\IbaseIRI$ and $\LTB$, and the (fixed) set $\symProjSet = \{\subjAttr, \predAttr, \objAttr%
			%, \graphAttr%
		\}$, $\fctProjectOpDflt{\mappingExpr_\mathsf{TM}}$ is an RML-spe\-cif\-ic mapping expression.

		\item For
			%every
		two RML-spe\-cif\-ic mapping expressions $\mappingExpr'\!$ and $\mappingExpr''\!$ such that $\mappingExpr''\!$ is of the form $\fctProjectOpDflt{\mappingExpr_\mathsf{TM}}$, $\fctUnion{\mappingExpr'\!}{\mappingExpr''}$ is an RML-spe\-cif\-ic mapping expression.
	\end{enumerate}
	For every RML-spe\-cif\-ic mapping expression $\mappingExpr$, we write $\fctTrmaps{\mappingExpr}$ to denote the set of all TrMap-ex\-pres\-sions contained in $\mappingExpr$.
\end{definition}

% \begin{example}
% 	\label{ex:RMLSpecificExpression}
% 	\inlinetodo{...}
% \end{example}

While the notions of an RML-spe\-cif\-ic mapping expression and of a TrMap-ex\-pres\-sion are defined with respect to an IRI~$\IbaseIRI$ and a function~$\LTB$, hereafter, we mention $\IbaseIRI$ and $\LTB$ only in cases in which they are explicitly relevant.

\subsection{Semantics of RML-Specific Mapping Expressions}
\label{ssec:Preliminaries:RMLSpecMappingExpressions:Semantics}

This section introduces a formal semantics of RML-spe\-cif\-ic mapping expressions.
	As a basis for evaluating
	%To evaluate
such an expression, it is necessary to assign concrete data objects to the source references mentioned in the $\extract$ operators of the expression, for which we introduce the notion of a source assignment.

\begin{definition}
	\label{def:SourceAssignment}
	\normalfont
	A \textbf{source assignment} is a partial function
	$\validInput \!: \symSourcerefUni \rightarrow \symDataObjUni$.
\end{definition}

Notice that such a source assignment may not be applicable to a given mapping expression. For instance, it may not cover all of the source references that occur within the expression or it may assign data objects that are not of the expected types. The following definition formalizes the conditions for a source assignment to be applicable,
	%both for a TrMap-ex\-pres\-sion and for an RML-spe\-cif\-ic mapping expression.
	for the types of expressions considered in this paper.

\begin{definition}
	\label{def:ValidInput}
	\normalfont
	A source assignment~$\validInput$ is a \textbf{valid input} for a TrMap-ex\-pres\-sion~$\mappingExpr_\mathsf{TM}$ if the following conditions hold:
	\begin{enumerate}
		\item
		If $\mappingExpr_\mathsf{TM}$ is of the form
		$
		\fctExtOpX{\objAttr}{\extExpr_\mathrm{o}}{
			\fctExtOpX{\predAttr}{\extExpr_\mathrm{p}}{
				\fctExtOpX{\subjAttr}{\extExpr_\mathrm{s}}{
					\fctExtractDflt
				}
			}
		}
		$
		with $\sourceType = \sourceTypeTuple$,
		then it must hold that
			%$\sourceReference \in \fctDom{\validInput}$ and
		$\fctValidInput{\sourceReference} \in \symSetDataset$.

		\item
		If $\mappingExpr_\mathsf{TM}$
			%is of the second form in Definition~\ref{def:TrMapExpression}; i.e., $\mappingExpr_\mathsf{TM}$
		is of the form
		\\ \hspace*{3mm}
		$\fctExtOpX{\objAttr}{\extExpr_\mathrm{o}'}{
			\fctEquiJoinOp{
				\fctExtOpX{\predAttr}{\extExpr_\mathrm{p}}{
					\fctExtOpX{\subjAttr}{\extExpr_\mathrm{s}}{
						\fctExtractDflt
					}
				}
			}{
				\fctExtractX{\sourceReference'}{\sourceType'\!}{\queryExpr'\!}{\symAttrQueryMap'}
			}
		}$
		\\ \hspace*{10mm} with $\sourceType = \sourceTypeTupleX{1}$ and
		\\ \hspace*{11mm} and $\sourceType'\! = \sourceTypeTupleX{2}$,
		\\ then it must hold that
			%$\sourceReference, \sourceReference'\! \in \fctDom{\validInput}$, $\fctValidInput{\sourceReference} \in \symSetDatasetX{1}$,
			$\fctValidInput{\sourceReference} \in \symSetDatasetX{1}$
		and $\fctValidInput{\sourceReference'} \in \symSetDatasetX{2}$.
	\end{enumerate}
	A source assignment~$\validInput$ is a \textbf{valid input} for an RML-spe\-cif\-ic mapping expression~$\mappingExpr$ if $\validInput$ is a valid input for every TrMap-ex\-pres\-sion in $\fctTrmaps{\mappingExpr}$.
\end{definition}

% \begin{example}
% 	\label{ex:ValidInput}
% 	\inlinetodo{...}
% \end{example}

Given the notion of valid inputs for evaluating RML-spe\-cif\-ic mapping expressions, we can now define the semantics of such an evaluation.

\begin{definition}
	\label{def:Semantics}
	\normalfont
	Let $\mappingExpr_\mathsf{rml}$ be an RML-spe\-cif\-ic mapping expression and
	$\validInput$ be a source assignment that is a valid input for $\mappingExpr_\mathsf{rml}$.
	For every sub-expression~$\mappingExpr$ of $\mappingExpr_\mathsf{rml}$, including $\mappingExpr_\mathsf{rml}$ itself, the \textbf{evaluation of $\mappingExpr$ based on $\validInput$}, denoted by $\fctApply{\mappingExpr}{\validInput}$, is the mapping relation $\mappingRel$ that is defined recursively as follows:
	\begin{enumerate}
		\item \label{def:Semantics:Extract}
		If $\mappingExpr$ is
			%of the form
		$\fctExtractDflt$ with $\sourceType = \sourceTypeTuple$, then
		      \begin{align*}
			      &\symAttrSubSet = \fctDom{\symAttrQueryMap} \quad \text{ and}
			      \\
			      &\symMappingInst = \big\{ \{\attr_1 \rightarrow \fctTypeCast{\var_1}, \dots, \attr_n \rightarrow \fctTypeCast{\var_n} \} \ \big| \
			       \dataObject \text{ is in } \symDataSeq \text{ and}             % Note that writing "is in" rather than $\in$ is deliberate because \symDataSeq is not a set but a sequence, for which \in is undefined. -Olaf
			      \\[-1mm]
			      & \hspace{67mm} ((\attr_1 , \var_1), \dots, (\attr_n , \var_n)) \in X_d \big\}
			      ,
			      \\[-6mm] % Layout Adjustment
		      \end{align*}
		      where $\symDataSeq = \fctEval{\bigDataObject}{\queryExpr}$ with $\bigDataObject = \fctValidInput{\sourceReference}$, and%
		      \vspace{-2mm} % Layout Adjustment
		      \begin{equation*}
			      X_d = \bigtimes_{\attr \in \fctDom{\symAttrQueryMap}}
			      \big\{ (\attr, \var) \ \big| \
			      \var \text{ is in } \fctCbeval{\bigDataObject}{\dataObject}{\queryExpr'} % Also here, writing "is in" instead of $\in$ is deliberate. -Olaf
			      \text{ with }  \queryExpr'\! = \symAttrQueryMap(\attr) \big\}
			      .
		      \end{equation*}

		\item \label{def:Semantics:Extend}
		If $\mappingExpr$ is
			of the form
		$\fctExtOp{\mappingExpr'}$,
		and given $\fctApply{\mappingExpr'}{\validInput} = (\symAttrSubSet'\!, \symMappingInst')$, then
		$$
		\symAttrSubSet = \symAttrSubSet' \cup \{\attr\}
		\quad \text{ and } \quad
		\symMappingInst = \bigl\{\mappingTuple \cup \{\attr \rightarrow \fctEval{\extExpr}{\mappingTuple}\} \ | \ \mappingTuple \in \symMappingInst' \bigr\}
		,
		$$
		where $\fctEval{\extExpr}{\mappingTuple}$ is the evaluation of
			%the extend expression~$\extExpr$ over mapping tuple~%
			$\extExpr$ over
		$\mappingTuple$, as per~\cite[Definition~9]{minoo2025AlgebraPublished}.

		\item \label{def:Semantics:Project}
		If $\mappingExpr$ is of the form $\fctProjectOpDflt{\mappingExpr'}$,
		and given $\fctApply{\mappingExpr'}{\validInput} = (\symAttrSubSet'\!, \symMappingInst')$, then
		$$
		\symAttrSubSet = \symAttrSubSet' \cap \symProjSet
		\quad \text{ and } \quad
		\symMappingInst = \{ \mappingTuple[\symAttrSubSet] \ | \ \mappingTuple \in \symMappingInst' \}
		,
		$$
		where $\mappingTuple[\symAttrSubSet]$
			%denotes the mapping tuple that is the restriction of $\mappingTuple$ to $\symAttrSubSet$.
			is the mapping tuple~$\mappingTuple'$ s.t.\ $\fctDom{\mappingTuple'}=\symAttrSubSet$, $\mappingTuple(\attr) = \mappingTuple'\!(\attr)$ for all $\attr \in \symAttrSubSet$.

		\item \label{def:Semantics:Join}
		If $\mappingExpr$ is
			%of the form
		$\fctEquiJoinOp{\mappingExpr'\!}{\mappingExpr''}$,
		and given $(\symAttrSubSet'\!, \symMappingInst') = \fctApply{\mappingExpr'}{\validInput}$
		and $(\symAttrSubSet''\!, \symMappingInst'') = \fctApply{\mappingExpr''}{\validInput}$,
		%then
		\begin{align*}
			 & \symAttrSubSet = \symAttrSubSet' \cup \symAttrSubSet''
			 \\
			 &\symMappingInst =  \{ \mappingTuple_1 \cup \mappingTuple_2 \ | \
			      \mappingTuple_1 \in \symMappingInst' \text { and }
			      \mappingTuple_2 \in \symMappingInst'' \text{ such that }
			      \mappingTuple_1(\attr_1) = \mappingTuple_2(\attr_2) \text{ for all }
			      (\attr_1,\attr_2) \in \symJoinAttrPairs
			      \}
			.
		\end{align*}

		\item \label{def:Semantics:Union}
		If $\mappingExpr$ is
			%of the form
		$\fctUnion{\mappingExpr'\!}{\mappingExpr''}$,
		and given $(\symAttrSubSet'\!, \symMappingInst') = \fctApply{\mappingExpr'}{\validInput}$
		and $(\symAttrSubSet''\!, \symMappingInst'') = \fctApply{\mappingExpr''}{\validInput}$,
			%then
		$$
		\symAttrSubSet = \symAttrSubSet' \cup \symAttrSubSet''
		\quad \text{ and } \quad
		\symMappingInst = \symMappingInst' \cup \symMappingInst''
		.
		$$
	\end{enumerate}
\end{definition}

\section{Satisfiability}
\label{sec:Satisfiability}
% \vspace{-2mm} % Layout Adjustment

Our pruning approach is based on a notion of satisfiability of triple patterns with respect to RDF data obtained via mappings. This section provides the relevant formal results, for which we begin by defining this notion of satisfiability.

%%% The following (commented) definition is more general than the one we need
%%% for the paper (which is given below). We may need the more general one in
%%% the future; e.g., for a journal version of the paper.   -Olaf 
% \begin{definition}
% 	\label{def:Satisfiability}
% 	\normalfont
% 	Let $\symPattern$ be a graph pattern and $\mappingExpr$ be a valid mapping expression for which it holds that $\{ \subjAttr, \predAttr, \objAttr\} \subseteq \fctSchema{\mappingExpr}$.
% 	$\symPattern$ is \textbf{satisfiable over} $\mappingExpr$ if there exists a source assignment $\validInput$ that is valid input for $\mappingExpr$ such that $\evalP{\symPattern}{\symRDFGraph} \neq \emptyset$ with $\symRDFGraph$ being the RDF graph resulting from the mapping relation $\fctApply{\mappingExpr}{\validInput}$.
% \end{definition}

\begin{definition}
	\label{def:Satisfiability}
	\normalfont
	Let $\mappingExpr$ be a TrMap-ex\-pres\-sion. % such that $\{ \subjAttr, \predAttr, \objAttr\} \subseteq \fctSchema{\mappingExpr}$.
	A triple pattern~$\symTP$ is \textbf{satisfiable over} $\mappingExpr$ if there exists a source assignment $\validInput$ that is valid input for $\mappingExpr$ such that $\evalP{\symTP}{\symRDFGraph} \neq \emptyset$
		%with $\symRDFGraph$ being
		where $\symRDFGraph$ is
	the RDF graph resulting from the mapping~relation~$\fctApply{\mappingExpr}{\validInput}$.
\end{definition}

% \begin{example}
% 	\label{ex:Satisfiability}
% 	\inlinetodo{small example that uses a TrMap-ex\-pres\-sion of an earlier example}
% \end{example}

The following result\footnote{The proofs of all formal results in this paper are provided in \ExtendedVersion{the Appendix}\PaperVersion{an extended version~\cite{ExtendedVersion}}.}
	%a separate PDF file that we have submitted as supplementary material, in addition to the paper itself.}
	%a separate PDF file that we have submitted as supplementary material and will put on arXiv after acceptance.}
	%, submitted as supplementary material. We will put it on arXiv after the review phase.}
is the first building block of our pruning approach. It shows that,
	%in the context of
	when
evaluating a SPARQL graph pattern over an RDF graph created by applying an RML-spe\-cif\-ic mapping expression, the correct query result may be produced without explicitly considering every TrMap-ex\-pres\-sion of the given
	%(RML-spe\-cif\-ic)
mapping expression. In particular, it is possible to ignore
	every TrMap-ex\-pres\-sion over which none of the triple patterns of the graph pattern is~%
	%TrMap-ex\-pres\-sions over which none of the triple patterns of the graph pattern are~%
satisfiable.

\begin{proposition}
	\label{prop:CorrectnessOfPruning}
	\normalfont
	Let $\symPattern$ be a graph pattern and let $\mappingExpr$ and $\mappingExpr'$ be RML-spe\-cif\-ic mapping expressions such that $\fctTrmaps{\mappingExpr'} \subseteq \fctTrmaps{\mappingExpr}$ and, for every TrMap-ex\-pres\-sion $\mappingExpr_\mathsf{TM} \in \bigl( \fctTrmaps{\mappingExpr} \setminus \fctTrmaps{\mappingExpr'} \bigr)$,
		%every triple pattern $\symTP$ in $\symPattern$ is \emph{not}
		\emph{none} of the triple patterns in $\symPattern$ is
	satisfiable over~$\mappingExpr_\mathsf{TM}$.
	Then, for every source assignment~$\validInput$ that is valid input for $\mappingExpr$%
		%(and, thus, also for $\mappingExpr'$)%
	, it holds that $\evalP{\symPattern}{\symRDFGraph} = \evalP{\symPattern}{\symRDFGraph'}$, where $\symRDFGraph$ and $\symRDFGraph'$ are the RDF graphs resulting from the mapping relations $\fctApply{\mappingExpr}{\validInput}$ and $\fctApply{\mappingExpr'}{\validInput}$, respectively.
\end{proposition}

% \proof
% ...
% \qed
% %
% \medskip

Based on Proposition~\ref{prop:CorrectnessOfPruning}, we may prune TrMap-ex\-pres\-sions from a given RML-spe\-cif\-ic mapping expression (when used in the context of evaluating graph patterns over the RDF graph produced by applying the mapping expression). To do so, however, we need to be able to determine which triple patterns are satisfiable over which TrMap-ex\-pres\-sions, which leads us to the following decision problem.

\medskip \noindent
\begin{tabular}{|p{\textwidth}|}
	\hline
	\textbf{Problem:} \textsc{Satisfiability(TPoverTrMap)} \\[0.5mm]
	Input: a triple pattern $\symTP$ and a TrMap-ex\-pres\-sion $\mappingExpr$
	\\
	Question: Is $\symTP$ satisfiable over $\mappingExpr$?
	\\
	\hline
\end{tabular}
\medskip

The following result shows that we cannot answer this question in general.

\begin{proposition}
	\label{prop:SatisfiabilityForTPsUndecidable}
	\normalfont
	\textsc{Satisfiability(TPoverTrMap)} is undecidable.
\end{proposition}

% \proof

% \vspace{-2mm} % Layout Adjustment
\section{Incompatibility} \label{sec:Incompatibility}
% \vspace{-2mm} % Layout Adjustment

While Proposition~\ref{prop:SatisfiabilityForTPsUndecidable} shows that there is no general way to determine, for any given triple pattern~$\symTP$ and any given TrMap-ex\-pres\-sion~$\mappingExpr$, whether $\symTP$ is satisfiable over~$\mappingExpr$%
	\ or not%
, there is a class of cases
	%(i.e., specific pairs of triple patterns and TrMap-ex\-pres\-sions)
for which we can at least determine whether a triple pattern is guaranteed to be \emph{not} satisfiable over a TrMap-ex\-pres\-sion~(which is indeed what we need to know to use Proposition~\ref{prop:CorrectnessOfPruning} for pruning).

In this section, we introduce a syntactic property to identify such cases and prove its correctness. We call this property \emph{incompatibility}; more specifically, we say that a triple pattern~$\symTP$ is \emph{incompatible with} a TrMap-ex\-pres\-sion~$\mappingExpr$ if the pair of $\symTP$ and $\mappingExpr$ has the property that we aim to introduce.
To define the property formally, we need two auxiliary concepts. The first of them are regular expressions that we construct from torb-extend expressions as follows.

\begin{definition}
	\label{def:RegEx}
	\normalfont
	Let $\extExpr$ be a torb-extend expression (as per Definition~\ref{def:TORBExtendExpression}).
	The \textbf{matching pattern of $\extExpr$}, denoted by $\mathrm{regex}(\extExpr)$, is a
        string to be used as a regular expression\footnote{We assume the extended regular expression notation of the POSIX standard:\\ {\fontsize{8}{10} \url{https://pubs.opengroup.org/onlinepubs/9799919799/basedefs/V1_chap09.html}}} and
		%regular expression that
	is defined recursively as follows:
	\begin{enumerate}
% 		\item If $\extExpr$ is an IRI, then the string $\mathrm{regex}(\extExpr)$ is
% 		      a version of this IRI in which every character
% 		      that has a special meaning in regular expressions is escaped with a
% 		      backslash.
		\item If $\extExpr$ is an RDF literal $\literalTuple$, then
		      the string $\mathrm{regex}(\extExpr)$ is a version of
		      $\lex$ in which every character that has a special meaning in
		      regular expressions is escaped with a backslash.
		\item If $\extExpr$ is an attribute, then $\mathrm{regex}(\extExpr)$ is the
		      string ".+".
		\item If $\extExpr$ is of the form $\fctConcat{\extExpr_1, \ldots}{\extExpr_n}$,
		      then $\mathrm{regex}(\extExpr)$ is the string obtained by concatenating the strings $\mathrm{regex}(\extExpr_1)$, ..., and $\mathrm{regex}(\extExpr_n)$, in this order.
	\end{enumerate}
\end{definition}

% \begin{example}
% 	\label{ex:RegEx}
% 		%Let $\extExpr_\mathsf{ex}$ be the torb-extend expression $\fctConcat{\ell}{\attr_1}$ as used by one of the $\extend$ operators of the TrMap-ex\-pres\-sion of Example~\ref{ex:TrMapExpression}. The matching pattern of $\extExpr_\mathsf{ex}$ is
% 		%The matching pattern of the torb-extend expression $\fctConcat{\ell}{\attr_1}$ with $\ell = (\ttl{"http://example.com/route/"}, \ttl{xsd:string})$, as used by one of the $\extend$ operators of the TrMap-ex\-pres\-sion of Example~\ref{ex:TrMapExpression}, is
% 		The matching pattern of the torb-extend expression $\fctConcat{\ell}{\attr_1}$ that is used by one of the $\extend$ operators of the TrMap-ex\-pres\-sion of Example~\ref{ex:TrMapExpression} is the string
% 	"\ttl{http:\textbackslash/\textbackslash/example.com\textbackslash/route\textbackslash/.+}".
% \end{example}

The second auxiliary concept that we need is a notion of incompatibility between IRIs and extend expressions, which is defined as follows.
\begin{definition}
	\label{def:IncompatibleIRI}
	\normalfont
	An IRI~$\symIRI$ is \textbf{incompatible with} an extend expression~$\extExpr$~\cite{minoo2025AlgebraPublished} if \emph{any} of the following
three
	properties holds:
	\begin{enumerate}
		\item $\extExpr$ is any RDF term but not the IRI~$\symIRI$.

% 		\item $\extExpr$ is of the form $\fctToLiteral{\extExpr'\!, \dt}$. %in which
% 		      %$\extExpr'$ is an extend expression
% 		      %and $\dt$ is an IRI.
% 
% 		\item $\extExpr$ is of the form $\fctToBNode{\extExpr'}$. %in which
% 		      %$\extExpr'$ is an extend expression
% 		      %and $\LTB$ is a function $\LTB \!: \symAllStrings \rightarrow \symAllBNodes$.

		\item $\extExpr$ is either of the form $\fctToLiteral{\extExpr'\!, \dt}$ or of the form $\fctToBNode{\extExpr'}$.

		\item $\extExpr$ is of the form $\fctToIRI{\extExpr'\!}{\IbaseIRI}$
		      with $\extExpr'$
			      being
		      a torb-extend expression
		      and $\IbaseIRI$ an IRI
		      such that i)~$\symIRI$ does not match
		      	the regular expression
				$\mathrm{regex}(\extExpr')$ and ii)~$\symIRI$ does not match the regular expression formed by prefixing $\mathrm{regex}(\extExpr')$ with
					a version of $\IbaseIRI$ in which  every character that has a special meaning in regular expressions is escaped with a backslash.
					%$\mathrm{regex}(\IbaseIRI)$.
	\end{enumerate}
\end{definition}

\begin{example}
	\label{ex:IncompatibleIRI}
	Let $\extExpr_\mathsf{ex}$ be the extend expression $\fctToIRI{ \fctConcat{\ell}{\attr_1} }{\IbaseIRI}$ as used by the outermost $\extend$ operator of the TrMap-ex\-pres\-sion in Example~\ref{ex:TrMapExpression}. Since the matching pattern of its torb-extend expression, $\fctConcat{\ell}{\attr_1}$, is the string "\ttl{http:\textbackslash/\textbackslash/example.com\textbackslash/route\textbackslash/.+}", the IRI \ttl{http://transit.api/route/43} is incompatible with $\extExpr_\mathsf{ex}$,
		%whereas the IRI \ttl{http://example.com/route/43} is not incompatible with $\extExpr_\mathsf{ex}$.
		%whereas the IRI \ttl{http://example.com/route/43} is not incompatible.
		whereas \ttl{http://example.com/route/43} is not incompatible.
\end{example}

% \begin{lemma}
% 	\label{lemma:IncompatibleIRI}
% 	\normalfont
% 	Let $\extExpr$ be an extend expression and let $\symIRI$ be an IRI that is incompatible with $\extExpr$. For every mapping tuple~$\mappingTuple$, it holds that $\fctEval{\extExpr}{\mappingTuple} \neq \symIRI$.
% \end{lemma}
% 
% \proof
% The lemma follows readily from Definitions~\ref{def:RegEx} and~\ref{def:IncompatibleIRI}, in combination with Definition~\ref{def:ExtendExpressions:Semantics} and the definition of the extension functions being used~($\toIRI$,
% 	$\texttt{toBNode}$,
% 	%$\toBNode$,
% $\toLiteral$, and indirectly $\concat$; all defined in \cite[Appendix~B]{minoo2025AlgebraPublished}).
% \qed

Now we are ready to define
	%the incompatibility property for pairs of triple patterns and TrMap-ex\-pres\-sions.
	%incompatibility for pairs of triple patterns and TrMap-ex\-pres\-sions.
	our main incompatibility property.

\begin{definition}
	\label{def:IncompatibleTP}
	\normalfont
		%Let $\mappingExpr$ be a TrMap-ex\-pres\-sion~(with some $\IbaseIRI$ and $\LTB$). A triple pattern $(s,p,o)$ is \textbf{incompatible with $\mappingExpr$}
		A triple pattern $(s,p,o)$ is \textbf{incompatible with} a TrMap-ex\-pres\-sion~$\mappingExpr$~(with some $\IbaseIRI$ and $\LTB$)
	if either of the following conditions holds.
	\begin{enumerate}
		\item $\mappingExpr$ is of the first
			of the two forms
			%form
		%
			%given
		in Definition~\ref{def:TrMapExpression}, i.e., $\mappingExpr$ is of the form:
		\begin{equation*}
% 			\fctExtOpX{\graphAttr}{\extExpr_\mathrm{g}}{
				\fctExtOpX{\objAttr}{\extExpr_\mathrm{o}}{
					\fctExtOpX{\predAttr}{\extExpr_\mathrm{p}}{
						\fctExtOpX{\subjAttr}{\extExpr_\mathrm{s}}{
							\fctExtractDflt
						}
					}
				}
% 			}
		\end{equation*}
		and \emph{any} of the following conditions holds:
		\begin{enumerate}
			\item \label{case:IncompatibleTP:1:1}
			$s$ is an IRI that is incompatible with $\extExpr_\mathrm{s}$.

			\item \label{case:IncompatibleTP:1:2}
			$p$ is an IRI that is incompatible with $\extExpr_\mathrm{p}$.

			\item \label{case:IncompatibleTP:1:3}
			$o$ is an IRI that is incompatible with $\extExpr_\mathrm{o}$.

			\item \label{case:IncompatibleTP:1:4}
			$o$ is a literal and $\extExpr_\mathrm{o}$ is of the form $\fctToIRI{\extExpr'}{\IbaseIRI}$.

			\item \label{case:IncompatibleTP:1:5}
			$o$ is a literal and $\extExpr_\mathrm{o}$ is of the form $\fctToBNode{\extExpr'}$.

			\item \label{case:IncompatibleTP:1:6}
			$o$ is a literal $\literalTuple$ and $\extExpr_\mathrm{o}$ is of the form $\fctToLiteral{\extExpr'\!, \dt'}$ such that $\lex$ does not match the regular expression $\mathrm{regex}(\extExpr')$ or $\dt \neq \dt'\!$.
		\end{enumerate}

		\item $\mappingExpr$ is of the second
			of the two forms
			%form
		%
			%given
		in Definition~\ref{def:TrMapExpression}, i.e., $\mappingExpr$ is of the form:
		\begin{equation*}
% 			\fctExtOpX{\graphAttr}{\extExpr_1}{
				\fctExtOpX{\objAttr}{\extExpr_\mathrm{o}'}{
					\fctEquiJoinOp{
						\fctExtOpX{\predAttr}{\extExpr_\mathrm{p}}{
							\fctExtOpX{\subjAttr}{\extExpr_\mathrm{s}}{
								\fctExtractDflt
							}
						}
					}{
						\fctExtractX{\sourceReference'}{\sourceType'\!}{\queryExpr'\!}{\symAttrQueryMap'}
					}
				}
% 			}
		\end{equation*}
		and \emph{any} of the following conditions holds:
		\begin{enumerate}
			\item \label{case:IncompatibleTP:2:1}
			$s$ is an IRI that is incompatible with $\extExpr_\mathrm{s}$.

			\item \label{case:IncompatibleTP:2:2}
			$p$ is an IRI that is incompatible with $\extExpr_\mathrm{p}$.

			\item \label{case:IncompatibleTP:2:3}
			$o$ is an IRI that is incompatible with $\extExpr_\mathrm{o}'$.

			\item \label{case:IncompatibleTP:2:4}
			$o$ is a literal.
		\end{enumerate}
	\end{enumerate}
\end{definition}

\begin{example}
	\label{ex:IncompatibleTP}
	Both triple patterns of
		the SPARQL query in
	Listing~\ref{lst:sparql_input} are incompatible with the TrMap-ex\-pres\-sion in Example~\ref{ex:TrMapExpression}. For the first triple pattern, the incompatibility is due to condition~1(c) of Definition~\ref{def:IncompatibleTP} (see also Example~\ref{ex:IncompatibleIRI}), and for the second triple pattern, it is due to conditions~1(b) and~1(d).
\end{example}

The following result shows that incompatibility implies unsatisfiability, which makes it
	%the final
	another
main building block of our pruning approach because it guarantees that we can rely on
	%(syntactic) incompatibility checks rather than undecidable satisfiability checks
	incompatibility checks
to make
	correct
pruning decisions.

\begin{proposition}
	\label{prop:IncompatibilityImpliesUnsatisfiability}
	\normalfont
	Let $\mappingExpr$ be a TrMap-ex\-pres\-sion%
		. For every triple pattern $\symTP$ that is incompatible with $\mappingExpr$, it holds that
		%\ and $\symTP$ be a triple pattern. If $\symTP$ is incompatible with $\mappingExpr$, then
	$\symTP$ is \emph{not} satisfiable over~$\mappingExpr$.
\end{proposition}

% \proof
% ...
% \qed
% 
% \medskip

Based on Proposition~\ref{prop:IncompatibilityImpliesUnsatisfiability} we can replace the satisfiability-based condition in Proposition~\ref{prop:CorrectnessOfPruning} by an incompatibility-based condition. That is, combining both propositions gives us the following result%
	, which shows the correctness of the pruning algorithm that we shall develop from it in the next section%
.

\begin{corollary}
	\label{cor:CorrectnessOfPruningBasedOnIncompatibility}
	\normalfont
	Let $\symPattern$ be a graph pattern and let $\mappingExpr$ and $\mappingExpr'$ be RML-spe\-cif\-ic mapping expressions such that $\fctTrmaps{\mappingExpr'} \subseteq \fctTrmaps{\mappingExpr}$ and, for every TrMap-ex\-pres\-sion $\mappingExpr_\mathsf{TM} \in \bigl( \fctTrmaps{\mappingExpr} \setminus \fctTrmaps{\mappingExpr'} \bigr)$, it holds that every triple pattern in $\symPattern$ is incompatible with~$\mappingExpr_\mathsf{TM}$.
	Then, for every source assignment~$\validInput$ that is valid input for $\mappingExpr$%
		%(and, thus, also for $\mappingExpr'$)%
	, it holds that $\evalP{\symPattern}{\symRDFGraph} = \evalP{\symPattern}{\symRDFGraph'}$, where $\symRDFGraph$ and $\symRDFGraph'$ are the RDF graphs resulting from the mapping relations $\fctApply{\mappingExpr}{\validInput}$ and $\fctApply{\mappingExpr'}{\validInput}$, respectively.
\end{corollary}

% \input{chapters/05_propositions.tex}
% \vspace{-3mm} % Layout Adjustment
\section{Pruning Algorithm}%
\label{sec:pruning_algorithm}
% \vspace{-2mm} % Layout Adjustment

%%% I have commented the following part as this should be in the preliminaries section.

% This section introduces the algorithm to prune RML mappings (specifically version 1.1.2~\cite{Dimou2024:RMLSpec}) using queries.
% Before we start with the actual algorithm, we introduce a few concepts which we
% will use throughout the rest of the section to define our pruning strategy.
% Using the definition of satisfiability of a graph pattern $\symPattern$ over
% mapping expression $\mappingExpr$ (Section~\ref{sec:satisfiability}),
% it is evident that,
% in the case of SPARQL queries,
% $\symPattern$ is equivalent to a \emph{triple pattern}.
% A triple pattern is a tuple
% \begin{equation*}
% 	(s,p,o) \in \subjectPatternTuple \times \predicatePatternTuple \times \objectPatternTuple,
% \end{equation*}
% where $\symAllVariables$ is the countably infinite set of all variables,
% that is a subset of $\symAllStrings$, and is disjoint from $\symAllIRIs, \symAllBNodes,\text{ and } \symAllLiterals$.
% Following SPARQL's convention, we write variables by beginning it with a question mark (i.e., \texttt{?s} is a variable).
% With these preliminaries in place, we now turn to the algorithm of query-based
% pruning of RML mappings for which the pseudocode is provided as Algorithm~\ref{alg:prune_rml}.

To capture our pruning approach in an algorithmic form, we convert Corollary~\ref{cor:CorrectnessOfPruningBasedOnIncompatibility} into Algorithm~\ref{alg:prune_rml}.
The input to this algorithm is a
	SPARQL
graph pattern~$\symPattern$ and an RML-spe\-cif\-ic mapping expression~$\mappingExpr$, and the output is an RML-spe\-cif\-ic mapping expression that
	%is constructed such that it corresponds to
	is constructed to become the
	%corresponds to
$\mappingExpr'$ in Corollary~\ref{cor:CorrectnessOfPruningBasedOnIncompatibility}.

\begin{algorithm}[b!]
	    \SetArgSty{textnormal}
    \DontPrintSemicolon
	\caption{Prunes an
		%RML-spe\-cif\-ic mapping expression
		RML mapping
	for a given SPARQL graph pattern.}
	\label{alg:prune_rml}
        \KwData{%
            $\symPattern$ - a SPARQL graph pattern,
            $\mappingExpr$ - an RML-spe\-cif\-ic mapping expression 
        }
        \KwResult{an RML-spe\-cif\-ic mapping expression}
        $R \gets \emptyset$ \tcp*[f]{will be used to collect the TrMap-ex\-pres\-sions that cannot be pruned\label{line:prune_rml:initR}}\; 
        \ForEach{TrMap-ex\-pres\-sion $\mappingExpr_\mathsf{TM} \in \fctTrmaps{\mappingExpr}$\label{line:prune_rml:start_prune_loop}}{
            \ForEach{triple pattern $\symTP$ in $\symPattern$\label{line:prune_rml:tpattern_forloop}}{
                \If{$\symTP$ is \emph{not} incompatible with $\mappingExpr_\mathsf{TM}$\label{line:prune_rml:tpattern_satisfiable}}{
                        $R \gets R \cup \big\{ \mappingExpr_\mathsf{TM} \big\}$ \label{line:prune_rml:project}\;
                        \Continue \tcp*[f]{break out of loop at line~\ref{line:prune_rml:tpattern_forloop}}
                    }
            }
            \label{line:prune_rml:end_prune_loop}
        }
        $\mappingExpr'\! \gets \fctProjectOp{\{\subjAttr,\predAttr,\objAttr%
            %,\graphAttr%
        \}\!}{ \mappingExpr_\mathsf{TM} }$, where $\mappingExpr_\mathsf{TM}$ is an arbitrary TrMap-ex\-pres\-sion of~$R$\label{line:prune_rml:combine_start}\; 
        \ForEach{TrMap-ex\-pres\-sion $\mappingExpr_\mathsf{TM}' \in R \setminus \{ \mappingExpr_\mathsf{TM} \}$}{
            $\mappingExpr' \gets \fctUnionBig{\mappingExpr'\!}{ \fctProjectOp{\{\subjAttr,\predAttr,\objAttr%
                %,\graphAttr%
            \}\!}{ \mappingExpr_\mathsf{TM}' } } $ 

            \label{line:prune_rml:combine_end}
        }

        \Return{$\mappingExpr'$}\label{line:prune_rml:return}

\end{algorithm}

The algorithm first collects the TrMap-ex\-pres\-sions of $\mappingExpr$ that cannot be pruned~(lines~\ref{line:prune_rml:initR}--\ref{line:prune_rml:end_prune_loop}). To this end, the algorithm iterates over all TrMap-ex\-pres\-sions of $\mappingExpr$. For each such TrMap-ex\-pres\-sion~$\mappingExpr_\mathsf{TM}$, the algorithm checks whether there is a triple pattern in $\symPattern$ that is \emph{not} incompatible with~$\mappingExpr_\mathsf{TM}$~(lines~\ref{line:prune_rml:tpattern_forloop}--\ref{line:prune_rml:tpattern_satisfiable}). As soon as the first such triple pattern is found, $\mappingExpr_\mathsf{TM}$ is collected~(line~\ref{line:prune_rml:project}) and the algorithm moves on to the next TrMap-ex\-pres\-sion of $\mappingExpr$. Hence, every TrMap-ex\-pres\-sion that has not been collected into $R$
	at the end of the for-loop in lines~\ref{line:prune_rml:start_prune_loop}--\ref{line:prune_rml:end_prune_loop}
	%after the for-loop in lines~\ref{line:prune_rml:start_prune_loop}--\ref{line:prune_rml:end_prune_loop} finishes
is one that every triple pattern of $\symPattern$ is incompatible with (and, thus, can be
	ignored as per Corollary~\ref{cor:CorrectnessOfPruningBasedOnIncompatibility}).
	%pruned).
%
After collecting the TrMap-ex\-pres\-sions
	%that need
to be kept, the algorithm reconstructs an RML-spe\-cif\-ic mapping expression with them~(lines~\ref{line:prune_rml:combine_start}--\ref{line:prune_rml:combine_end}).

\section{Evaluation}
\label{sec:evaluation}
% \vspace{-3mm} % Layout Adjustment

\begin{figure}[b]
	\centering
	\includegraphics[width=0.8\textwidth]{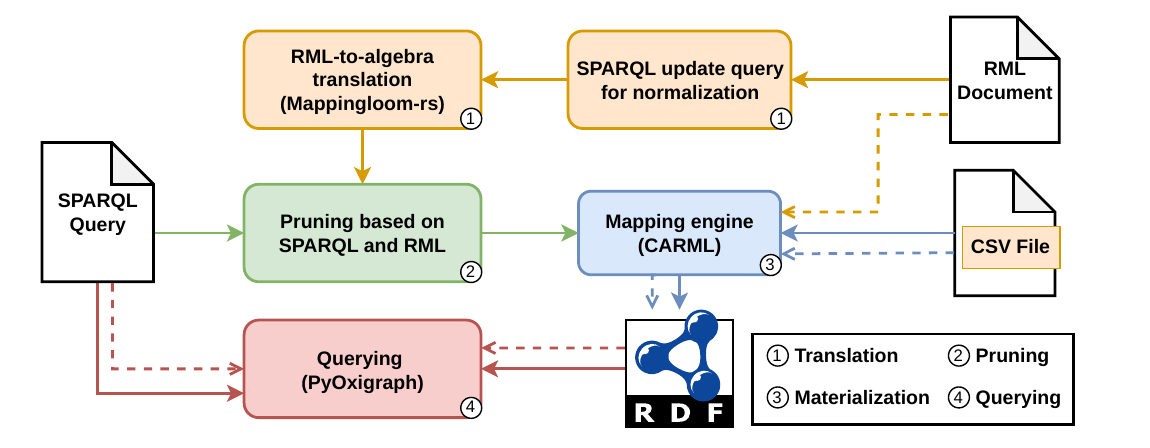}
% 	\vspace{-4mm} % Layout Adjustment
	\caption{The evaluation pipeline consists of four steps:
		i)~translation%
			%~(including an initial normalization of the RML mapping)%
		,
		ii)~pruning,
		iii)~materialization, and
		iv)~querying.
		The solid-lined arrows represent the execution flow that applies
		%our pruning approach,
		our query-spe\-cif\-ic pruning approach,
		whereas dashed arrows capture the baseline execution flow without pruning.}
	\label{fig:evaluation_pipeline}
% 	\vspace{3mm} % Layout Adjustment
\end{figure}

	%In this section, we
We now
evaluate how pruning of RML mappings
affects the execution time of each step of the evaluation pipeline shown in
Figure~\ref{fig:evaluation_pipeline}.
%In this work, we focus on measuring the isolated performance impact of our pruning approach.
Specifically, we conduct a self-evaluation where, given a SPARQL query and an RML mapping, 
we first prune the RML mapping for the query and then run the pruned RML mapping on a chosen mapping engine,
instead of comparing the performance across multiple mapping engines%
	%~\cite{ArenasGuerrero2022MorphKGCScalable,Iglesias2022ScalingKnowledgeGraph,MinOo2022RMLStreamer,Freund2024FlexRML}%
.
Self-evaluation enables us to isolate the performance impact of the pruning.
Moreover, by using an existing mapping engine, we show that such engines can
benefit from the performance improvements of our approach.

We
	%also emphasize that we
chose not to compare against virtualization systems for the following reason:
	%Such systems not only select relevant mapping rules but, additionally, they
	In addition to selecting relevant mapping rules, such systems
rewrite the given SPARQL query to the query language
	%supported by
	of
the underlying source database system, which in turn uses rewritten and optimized query plans to retrieve relevant data. Such a cascading application of optimizations makes it difficult to isolate the performance impact of the pruning strategy used by~these~systems.

	%In this section, we
	In the following subsections, we
	%We
first introduce the chosen benchmark and our rationale for choosing it~(Section~\ref{sub:benchmark_config}).
	%We then
	Thereafter, we
describe our evaluation setup~(Section~\ref{sub:eval_setup}), and present and analyze the results~(Section~\ref{sub:eval_results}).

\subsection{Benchmark}
\label{sub:benchmark_config}
For our evaluation, we use the GTFS-Madrid benchmark~\cite{Chaves2020GTFSMadridBench}.
It can generate synthetic source datasets based on real-world data about the Madrid subway network, scaled by a factor, and comes with 18 SPARQL queries
based on actual user queries.
The benchmark task is to answer these queries over the RDF representation of the generated dataset, referred to as the \emph{full RDF graph} in the rest of this section.
The benchmark includes an RML mapping document containing 86 TrMap-expressions to generate the full RDF graph.
The provided SPARQL queries are diverse, using SPARQL features such as \ttl{FILTER}, \ttl{OPTIONAL}, and \ttl{GROUP BY}, and contain 3--15 triple patterns each~\cite{Chaves2020GTFSMadridBench}.
While the benchmark is mainly used to evaluate virtualization engines~\cite{Chaves2020GTFSMadridBench},
it has also been adapted for evaluating materialization engines~\cite{ArenasGuerrero2022MorphKGCScalable}, making it relevant for our approach.
For each of the benchmark queries, we prune the benchmark RML mapping to generate an RDF graph that is a subset of the full RDF graph.
We configure the GTFS-Madrid benchmark to generate the synthetic dataset at scale factor~10.
This small scale factor fits our aim, which is to evaluate the performance impact of our pruning technique
rather than the scalability of an end-to-end data-mapping pipeline.
Furthermore, to limit the impact of a potential implementation error in the underlying mapping engine when parsing complex input data, 
we configure the benchmark to generate data in CSV format.

% \vspace{-3mm} % Layout Adjustment
\subsection{Setup}%
\label{sub:eval_setup}

Figure~\ref{fig:evaluation_pipeline} illustrates our evaluation pipeline.
For each step of the pipeline, we measure the time needed to complete
the step.
To ensure that these time measurements can be clearly isolated from one another, we have implemented each step of the pipeline via a separate component.
First, the translation step normalizes the RML document of the benchmark using SPARQL update queries and, then,
translates the normalized RML mappings into the mapping algebra, both as defined
	in our earlier work%
	%by Min Oo and Hartig%
~\cite{minoo2025AlgebraPublished} (where the normalization step is an adaptation of normalizations by Kontchakov et al.~\cite{Kontchakov2014SPARQLanswerOWL} and by Rodr\'iguez-Muro and Rezk~\cite{Rordriguez2015SPARQL2SQL}).
The time required for this step is measured as the \emph{translation time}. 
%anonymize done
The resulting (RML-spe\-cif\-ic) mapping expression is then pruned by using our approach of Section~\ref{sec:pruning_algorithm}, separately for
	each of the 18~queries
	%every query
of the benchmark.
	The time required for this process per query 
	%Execution time for the pruning
is measured as the \emph{pruning time}. 
For every query, we then create an RML mapping, containing each of the remaining TrMap-expressions as
	%an
	a separate
RML triples map,
and serialize it into a file
	%to be used
as the input for the next step.
	%In the following \emph{materialization step},
	%Next,
	In that step,
the pruned RML mappings are executed using the CARML\footnote{\url{https://github.com/carml/carml-jar/releases/tag/v1.4.0}} mapping engine, and
the resulting partial RDF graphs are serialized into files, one per query.
The time taken to materialize each such RDF graph is measured as the \emph{materialization time} (per query). 
	%Finally, the querying step consists of two sub-steps. First, the file with the
	Next, for every query, the file with the corresponding
RDF graph from the materialization step is loaded into an
Oxigraph\footnote{\url{https://github.com/oxigraph/oxigraph}} triple store.
	%We exclude the loading step from measurements. Second,
	Finally,
we conduct the actual \emph{querying time} evaluation by running the
SPARQL query under consideration on the loaded RDF graph, with a Python wrapper
for Oxigraph.\footnote{\url{https://pypi.org/project/pyoxigraph/}} 
%%% anonymization - replace after acceptance

To establish a baseline for our evaluation, we measure the time taken
with the original RML mapping file of the benchmark to
\begin{enumerate*}[label=\roman*)]
	\item execute it
		%via CARML
	to generate the full RDF graph without pruning (\emph{baseline materialization time}), and
	\item query the full RDF graph with each of the 18 SPARQL queries (\emph{baseline querying times}).  
\end{enumerate*}
For these
	%baseline
measurements, % of execution time on materialization and querying, 
the execution follows the dashed path in Figure~\ref{fig:evaluation_pipeline}.

The evaluation is conducted on a machine with an Intel i7 CPU~(4.80 GHz) and
16~GB of RAM, running Ubuntu 24.04.3.
The whole pipeline is executed five times for each of the 18 SPARQL queries,
with the individual steps executed sequentially as shown in Figure~\ref{fig:evaluation_pipeline}.
The first run is used to warm up the pipeline and its measurements are discarded;
the measurements of the remaining four runs are averaged per step per query.
Table~\ref{tab:result} presents these measurements.

% \vspace{-3mm} % Layout Adjustment
\subsection{Results}%
\label{sub:eval_results}

% \vspace{-1mm} % Layout Adjustment

\subsubsection{Pruning Time.}%
\label{ssub:pruning_time}

	We observe that the
	%The
time
	%required
	used
for pruning is negligible: between 3 and 5~ms, which is
	%several
orders of magnitude shorter than even the shortest materialization time
	with one of the pruned mappings
(1,717.52~ms%
	%), and also shorter than the translation time~(218.22~ms)%
	,~for~Q6)%
.

% \begin{table}[h!]  %% Don't use [h] for tables, figures, etc., because that way there is always the need to have some space both above *and* below the table/figure. If you use [t], there needs to be space only below the table/figure; if you use [b], there needs to be space only above the table/figure!!
\begin{table}[t]
\vspace{-5mm}
	\caption{For each query, the execution time for different steps of the evaluation pipeline is averaged over four runs.
	The last column reports the querying time when querying the full RDF graph.
    \texttt{TO} means the execution timed out (3600s). 
	}\label{tab:result}
	\centering
    \resizebox{\textwidth}{!}{
		\begin{tabular}{|c|r|r|r|r|r|r||r|}
			\hline
            \multicolumn{1}{|m{2cm}|}{\centering Query} & \multicolumn{1}{m{2cm}|}{\centering Pruning} & \multicolumn{1}{m{2.5cm}|}{\centering \# of TrMaps \\after pruning} &\multicolumn{1}{m{2.2cm}|}{\centering \# of Triples generated}&\multicolumn{1}{m{2.5cm}|}{\centering Materialization}&\multicolumn{1}{m{4cm}|}{\centering Pruning + Materialization as \% of Baseline Materialization } & \multicolumn{1}{m{1.8cm}||}{\centering Querying}& \multicolumn{1}{m{1.8cm}|}{\centering Querying \\without\\ pruning}  \\
			\hline
            Q1                 & 3.39 ms   & 7      & 2,952,240     & 81,474.83 ms &    99.46 \%  & 0.09 ms  & 0.11 ms  \\
			\hline
            Q2                 & 3.40 ms   & 9      & 1,235,330     & 15,349.36 ms &   18.74 \%  & 9.62 ms   & 11.71 ms  \\
			\hline
            Q3                 & 3.55 ms   & 10     & 1,247,950     & 15,092.22 ms &   18.42 \%  & 0.13 ms   & 0.12 ms   \\
			\hline
            Q4                 & 3.28 ms   & 14     & 39,730        & 2,822.38 ms  &   3.45 \%   & 4.15 ms   & 5.64 ms   \\
			\hline
            Q5                 & 4.01 ms   & 7      & 3,600         & 2,044.86 ms  &   2.50 \%   & 0.08 ms   & 0.15 ms   \\
			\hline
            Q6                 & 2.88 ms   & 2      &  260          & 1,717.52 ms  &   2.10 \%   & 0.21 ms   & 0.35 ms   \\
			\hline
            Q7                 & 4.63 ms   & 20     & 1,321,430     & 17,362.18 ms &    21.20 \%  & \texttt{TO }        & \texttt{TO } \\
			\hline
            Q8                 & 4.29 ms   & 18     & 112,970       & 6,150.15 ms  &   7.51 \%    & \texttt{TO }         & \texttt{TO }  \\
			\hline
            Q9                 & 4.12 ms   & 11     & 1,776,620     & 68,753.60 ms &    83.93 \%  & 0.41 ms  & 0.62 ms  \\
			\hline
            Q10                & 3.28 ms   & 5      & 80,770         & 3,889.52 ms &   4.75 \%   & 43.59 ms   & 78.56 ms \\
			\hline
            Q11                & 4.48 ms   & 14     & 6,370         & 2,175.37 ms  &   2.66 \%   & 0.12 ms   & 0.18 ms  \\
			\hline
            Q12                & 4.32 ms   & 13     & 121,500       & 4,606.08 ms  &   5.62 \%   & 72.26 ms   & 125.79 ms \\
			\hline
            Q13                & 3.24 ms   & 5      & 45,820         & 3,007.90 ms &   3.67 \%   & 8.37 ms   & 22.68 ms  \\
			\hline
            Q14                & 3.43 ms   & 10     & 129,660       & 4,824.48 ms  &   5.89 \%   & 99.14 ms   & 165.38 ms  \\
			\hline
            Q15                & 4.46 ms   & 86     & 4,546,610     & 102,905.99 ms &  125.63 \% & 0.38 ms    & 0.65 ms  \\
            \hline
            Q16                & 4.60 ms   & 11     & 8,800         & 2,231.85 ms  &    2.72 \%  & 0.09 ms    & 0.17 ms  \\
			\hline
            Q17                & 3.69 ms   & 11     & 62,130        & 4,390.00 ms  &     5.36 \%  & 0.12 ms   & 0.14 ms  \\
			\hline
            Q18                & 4.46 ms   & 12     & 7,060         & 2,120.87 ms  &     2.59 \%  & 0.11 ms   & 0.13 ms  \\
			\hline
			\multicolumn{8}{l}{} 
			\\
				\multicolumn{2}{p{3.9cm}}{Measurement of translation time, which is query independent (Step 1 in Fig.~\ref{fig:evaluation_pipeline})} &
				\multicolumn{4}{p{9cm}}{Baseline measurements for the number of TrMap-expressions in the original RML document, 
            the number of triples generated and the materialization time without pruning nor translation} &
				\multicolumn{2}{c}{}
         \\ \cline{1-1}\cline{3-5}
				\multicolumn{1}{|m{2cm}|}{\centering Translation} &
				\multicolumn{1}{c}{} &
            \multicolumn{1}{|m{2.5cm}|}{\centering \# of TrMaps \\before pruning} &
            \multicolumn{1}{m{2.2cm}|}{\centering \# of Triples generated} &
			\multicolumn{1}{m{2.5cm}|}{\centering Baseline Materialization \\ w/o pruning} &
			\multicolumn{3}{c}{}
			\\ \cline{1-1}\cline{3-5}
			\multicolumn{1}{|c|}{218.22 ms} &
			\multicolumn{1}{c}{} &
			\multicolumn{1}{|r|}{86} &
			\multicolumn{1}{r|}{4,546,610} &
			\multicolumn{1}{r|}{81,915.27 ms} &
			\multicolumn{3}{c}{}
			\\ \cline{1-1}\cline{3-5}
		\end{tabular}
	}
% \vspace{-3mm} % Layout Adjustment
\end{table}

\vspace{-3mm} % Layout Adjustment
\subsubsection{Materialization Time.}%
\label{ssub:materialization_time}
	%We observe that, for
	For
12 of the 18 queries, our approach reduces the
materialization time after pruning to at most 8\% (around 7~secs) of the baseline materialization time (82~secs)!
Such significant improvement is due to the small number of remaining TrMap-expressions
after pruning.
Specifically, the original 86 TrMap-expressions are reduced to at most 20, which significantly lowers the computation effort of the
materialization step~(ignoring Q15 for the moment).

	%However,
	We also observe that,
for Q1 and Q9, even though the pruning step retains fewer
TrMap-expressions than for Q7, their materialization times are higher than for Q7,
which we explain as follows. 
These two queries contain the triple pattern \ttl{"?shape gtfs:shapePoint
?shapePoint."} Triples that match this triple pattern are produced by a
TrMap-ex\-pres\-sion that is of the second form in
Definition~\ref{def:TrMapExpression} (i.e., using a join) and its two source
references are the same. 
Hence, evaluating this expression results in a self-join operation, which is
known to cause higher materialization
time~\cite{deVleeschauwer2024RML-view-to-CSV}.
To make matters worse, the source reference refers to the  
largest file of the input dataset, \ttl{SHAPES.csv}, which
contains about 585K records at the considered scale factor.
%Evaluating a self join on the largest file results in a materialization time
%that is significantly higher compared to the other cases.
Despite such limitations, our approach still achieves lower materialization
times for Q1 and Q9 than the baseline.

Materializing the RDF graph for Q15 took longer than the baseline materialization time (about 126\%).
Q15 contains the triple pattern \ttl{?stop ?p ?str}, which means that no TrMap-ex\-pres\-sions are pruned.
Due to the normalization step, each RML triples map that contains multiple
predicate-object maps is normalized into several separate RML triples maps,
each with one predicate and one object map~\cite{minoo2025AlgebraPublished}. 
As a result, the re-generated RML mapping, using the retained TrMap-expressions,
contains more RML triples map definitions than the original RML mapping. 
However, they are semantically the same (i.e., evaluating them independently
produces the same RDF graph).
A mapping engine that is not aware of such semantic equivalence of two
differently-written RML mappings, may take longer for one than for the other.
We can infer from our measurements that CARML is not aware of such semantic equivalence.

\vspace{-3mm} % Layout Adjustment
\subsubsection{Querying Time.}%
\label{ssub:querying_time}
Our measurements
	%in Table~\ref{tab:result}
show that querying the RDF graphs generated via the pruned RML mappings
is often faster than querying the full RDF graph.
The speed up is especially noticeable for Q10, Q12, and Q14,
where the querying time is smaller by around 44\%,
42\%, and 40\%, respectively.

For Q7 and Q8, with and without pruning, the query execution timed out at 3600~seconds.
Q7 and Q8 contain the most triple patterns, 15 and 14 respectively,
while also using the \ttl{OPTIONAL} feature of SPARQL and, for Q7, even \ttl{DISTINCT}.
Although Q10 also uses \ttl{DISTINCT}, its querying time is substantially smaller because
the corresponding generated RDF graph is much smaller (around 80K triples, compared 1.3M triples for Q7).
For Q8, the timeout occurred due to the combination of having the second most
triple patterns (14) and the largest chained star-shaped group, both of which
increases querying complexity.

% \vspace{-3mm} % Layout Adjustment
\section{Related Work}%
\label{sec:related_works}
% \vspace{-2mm} % Layout Adjustment
As mentioned in the introduction, there are two types of approaches to provide access to
RDF views of non-RDF data:
\emph{materialization} and \emph{virtualization}.

Materialization approaches rely on a mapping language (e.g., RML~\cite{dimou_ldow_2014} or R2RML~\cite{Seema2012R2RML}), define a direct mapping using a meta-model~\cite{asprino2023knowledge}, or use manually written scripts to transform non-RDF data to generate the RDF graph.
Thus, materialization systems~\cite{ArenasGuerrero2022MorphKGCScalable,asprino2023knowledge,Iglesias2022ScalingKnowledgeGraph,deVleeschauwer2024RML-view-to-CSV}
can process diverse types of data sources and formats into RDF graph, 
provided that the types of data sources and formats are supported by the
underlying mapping language, meta-model or implementation.
However, none of these systems considers a user's query while generating an RDF graph, which can be inefficient in dynamic environments where a query can be
answered by a small subset of the generated RDF graph.

Virtualization approaches provide SPARQL query functionality over a virtual RDF view of the
underlying non-RDF data.
To this end, these approaches rewrite any given SPARQL query into an equivalent query supported by the underlying data source.
The rewritten query is then executed against the data source, and the fetched results
are transformed back into SPARQL query results.
Such query rewriting and result transformation is usually guided by the usage of a mapping language such as R2RML~\cite{Seema2012R2RML}.
Thus, virtualization systems~\cite{calvanese2017ontop,priyatna2014formalization,ChavesFraga2021Enhancingvirtualontology}
are capable of answering SPARQL queries over non-RDF data through virtual RDF views.  
However, query rewriting also limits such systems to using data sources that
support a query language into which a SPARQL query can be translated.

While virtualization is query-aware, it is less flexible than materialization
in supported data sources, whereas materialization can handle diverse sources
but generates unnecessary triples.
Our approach bridges this gap by preemptively pruning
mappings, to retain mappings relevant to answering the user’s query.

% \vspace{-3mm} % Layout Adjustment
\section{Concluding Remarks and Future Work}%
\label{sec:concluding_remarks}
% \vspace{-2mm} % Layout Adjustment
The approach presented in this work opens up new possibilities for research on using non-RDF data sources when evaluating SPARQL queries over a federation of data sources. 
More concretely, there are two scenarios in which our pruning approach is beneficial. 
First, if a non-RDF data source is wrapped as a SPARQL endpoint to provide access to an RDF view of the underlying data, generating an up-to-date (and pruned) version of this view at query time is particularly relevant if the source data changes frequently.
Second, a query federation engine that can query non-RDF data sources through materialization at runtime~\cite{Hartig26:HeFQUINWebAPIsDemo} will benefit from the pruning of irrelevant mappings. 
As an example of the latter case, consider a JSON-based REST API as such a data source: if a part of the SPARQL query over the whole federation is
meant to be matched in the RDF views of the data from requests to that API, the
pruning approach will speed up the process of both materializing the data
retrieved from the API and evaluating the relevant sub-pattern of the SPARQL
query over this materialized RDF view. 
%
%{\color{red} \sout{Our pruning approach enables existing materialization engines to be used (as demonstrated with CARML)
%as SPARQL endpoints that materialize only the required RDF graph to answer the query on-the-fly. 
%As a result, this reduces the need for always-on, computationally intensive triple stores for query answering, 
%and thus, lowering operational costs and improving the flexibility of federated querying architecture. }}

While our pruning approach works at the triple pattern level, a natural next step is to 
extend it to whole basic graph patterns: Performing satisfiability checks that consider joins between multiple triple patterns 
may result in pruning even more TrMap-expressions and, thus, reduce the materialization time even further.

\enlargethispage{\baselineskip}%  % Layout Adjustment
\begin{credits}
	\subsubsection{\ackname}
	% A bold run-in heading in small font size at the end of the paper is
	% used for general acknowledgments, for example: This study was funded
		%The work presented in this paper
		%The presented work
		This work
	was supported
	by the Knut and Alice Wallenberg Foundation (KAW 2023.0111),
	by
		%Vetenskapsrådet (the Swedish Research Council,
		the Swedish Research Council (%
		project
		%prj.\
	reg.\ no.\ 2025-06246),
	and
	by the imec.icon project PACSOI (HBC.2023.0752), which was co-financed by imec and VLAIO and brings together the following partners: FAQIR Foundation, FAQIR Institute, MoveUP, Byteflies, AContrario, and Ghent University~–~IDLab.
 
 \vspace{-3mm} % Layout Adjustment
 	\subsubsection{\discintname}
 	% It is now necessary to declare any competing interests or to specifically
 	% state that the authors have no competing interests. Please place the
 	% statement with a bold run-in heading in small font size beneath the
 	% (optional) acknowledgments\footnote{If EquinOCS, our proceedings submission
 	% system, is used, then the disclaimer can be provided directly in the system.},
 	% for example: The authors have no competing interests to declare that are
 	% relevant to the content of this article. Or: Author A has received research
 	% grants from Company W. Author B has received a speaker honorarium from
 	% Company X and owns stock in Company Y. Author C is a member of committee Z.
 	%
 	The authors have no competing interests to declare%
 	%\ that are relevant to the content of this article%
 	.
 \end{credits}

\paragraph*{Supplemental Material Statement:}
The source code and artifacts for the evaluation, including an implementation of the pruning algorithm, are provided in a GitHub
repository.\footnote{\url{https://github.com/s-minoo/satisfiability-experiment}}
Full proofs of Propositions~\ref{prop:CorrectnessOfPruning}%
  %, \ref{prop:SatisfiabilityForTPsUndecidable}, and
  --%
\ref{prop:IncompatibilityImpliesUnsatisfiability}
  are in \ExtendedVersion{the Appendix}\PaperVersion{\cite{ExtendedVersion}}.
  %are attached with the submission on EasyChair.
%%and, if accepted, will be published on arXiv in an extended version of the paper.

\paragraph*{Declaration of use of Generative AI:}
The authors have not employed any Generative AI tools for the presented work.

\bibliographystyle{splncs04}
\bibliography{bibliography}

\ExtendedVersion{
	% \newpage

% \begin{center}
% \bfseries \Large
% Supplementary Material for the Paper:
% \\
% "Query-Specific Pruning of RML Mappings"
% \end{center}

\appendix
\section*{\appendixname}
\subsection*{Proof of Proposition~\ref{prop:CorrectnessOfPruning}}
\addcontentsline{toc}{section}{Proof of Proposition~\ref{prop:CorrectnessOfPruning}}

Before we focus directly on proving Proposition~\ref{prop:CorrectnessOfPruning}, we first show the following two observations which shall become relevant in the proof of the proposition.

\begin{lemma} \label{lemma:SubGraphs2}
% \begin{note} \label{note:SubGraphs2}
	\normalfont
	Let $\mappingExpr$ and $\mappingExpr'$ be RML-spe\-cif\-ic mapping expressions such that
		%we have
	$\fctTrmaps{\mappingExpr'} \subseteq \fctTrmaps{\mappingExpr}$,
	let $\validInput$ be a source assignment that is valid input for $\mappingExpr$%
		%(and, thus, also for $\mappingExpr'$)%
	, and let $\symRDFGraph$ and $\symRDFGraph'$ be the RDF graphs resulting from the mapping relations $\fctApply{\mappingExpr}{\validInput}$ and $\fctApply{\mappingExpr'}{\validInput}$, respectively%
		%~(as per Definition~\ref{def:ResultingRDFGraph})%
	.
	Then,
		$\validInput$ is a valid input also for $\mappingExpr'\!$, and
		it holds that~%
	$\symRDFGraph'\! \subseteq \symRDFGraph$.
\end{lemma}
% \end{note}

\proof
The observation that $\validInput$ is a valid input also for $\mappingExpr'$ is a direct consequence of Definitions~\ref{def:ValidInput} and~\ref{def:RMLSpecificExpression}, in combination with the fact that $\fctTrmaps{\mappingExpr'} \subseteq \fctTrmaps{\mappingExpr}$.
To see that $\symRDFGraph'\! \subseteq \symRDFGraph$, let $\fctApply{\mappingExpr}{\validInput} = \mappingRel$ and $\fctApply{\mappingExpr'}{\validInput} = (\symAttrSubSet'\!, \symMappingInst')$. Then, by Definition~\ref{def:RMLSpecificExpression} and Definition~\ref{def:Semantics} (in particular, case~\ref{def:Semantics:Union}), it holds that $\symMappingInst'\! \subseteq \symMappingInst$. Considering the latter in the context of Definition~\ref{def:ResultingRDFGraph}, it follows that $\symRDFGraph'\! \subseteq \symRDFGraph$. \qed

\begin{lemma} \label{lemma:SubGraphs1}
% \begin{note} \label{note:SubGraphs1}
	\normalfont
	Let $\mappingExpr$ be an RML-spe\-cif\-ic mapping expression,
	let $\mappingExpr'\! \in \fctTrmaps{\mappingExpr}$ be a TrMap-ex\-pres\-sion in $\mappingExpr$,
	let $\validInput$ be a source assignment that is valid input for $\mappingExpr$%
		%(and, thus, also for $\mappingExpr'$)%
	, and let $\symRDFGraph$ and $\symRDFGraph'$ be the RDF graphs resulting from the mapping relations $\fctApply{\mappingExpr}{\validInput}$ and $\fctApply{\mappingExpr'}{\validInput}$, respectively%
		%~(as per Definition~\ref{def:ResultingRDFGraph})%
	.
	Then,
		$\validInput$ is a valid input also for $\mappingExpr'\!$, and
		it holds that
	$\symRDFGraph'\! \subseteq \symRDFGraph$.
\end{lemma}
% \end{note}

\proof
Lemma~\ref{lemma:SubGraphs1} can be seen as a special case of Lemma~\ref{lemma:SubGraphs2} in which $\fctTrmaps{\mappingExpr'}$ (in Lemma~\ref{lemma:SubGraphs2}) is a singleton set.
\qed

\medskip

Now we prove Proposition~\ref{prop:CorrectnessOfPruning}:
Let $\validInput$ be a source assignment that is valid input for $\mappingExpr$, and thus also for $\mappingExpr'$%
	%~(which follows from Definitions~\ref{def:ValidInput} and~\ref{def:RMLSpecificExpression}, in combination with the fact that $\fctTrmaps{\mappingExpr'} \subseteq \fctTrmaps{\mappingExpr}$)%
	~(see Lemma~\ref{lemma:SubGraphs2})%
	%~(see Note~\ref{note:SubGraphs2})%
.
Furthermore, let $\symRDFGraph$ and $\symRDFGraph'$ be the RDF graphs resulting from the mapping relations $\fctApply{\mappingExpr}{\validInput}$ and $\fctApply{\mappingExpr'}{\validInput}$, respectively.
We have to show that $\evalP{\symPattern}{\symRDFGraph} = \evalP{\symPattern}{\symRDFGraph'}$, which we do by induction on the structure~of~$\symPattern$.

\emph{Base case:} Suppose $\symPattern$ is a triple pattern~$\symTP$.
	%In this case, to
	To
show that $\evalP{\symTP}{\symRDFGraph} = \evalP{\symTP}{\symRDFGraph'}$, we first show that $\evalP{\symTP}{\symRDFGraph} \subseteq \evalP{\symTP}{\symRDFGraph'}$%
	%\ and, thereafter, we show that $\evalP{\symTP}{\symRDFGraph} \supseteq \evalP{\symTP}{\symRDFGraph'}$. \par To show that $\evalP{\symTP}{\symRDFGraph} \subseteq \evalP{\symTP}{\symRDFGraph'\!}$,
	, for which
we let $\mu$ be a solution mapping in $\evalP{\symTP}{\symRDFGraph}$ and show that $\mu \in \evalP{\symTP}{\symRDFGraph'}$:
Since $\mu \in \evalP{\symTP}{\symRDFGraph}$, we know that $\symTP$ is satisfiable over $\mappingExpr$ and, thus, there exists a TrMap-ex\-pres\-sion~$\mappingExpr''\! \in \fctTrmaps{\mappingExpr}$ such that $\symTP$ is satisfiable over $\mappingExpr''$ and there exists a triple~$t$ in the RDF graph~$\symRDFGraph''$ resulting from $\fctApply{\mappingExpr''}{\validInput}$ such that $\mu[\symTP] = t$. Given that $\symTP$ is satisfiable over $\mappingExpr''\!$, $\mappingExpr''$ cannot be in $\fctTrmaps{\mappingExpr} \setminus \fctTrmaps{\mappingExpr'}$%
	~(because $\fctTrmaps{\mappingExpr} \setminus \fctTrmaps{\mappingExpr'}$ contains only TrMap-ex\-pres\-sions over which all triple patterns in $\symPattern$---which is $\symTP$ in this case---are \emph{not} satisfiable)%
. Therefore, $\mappingExpr''$ must be in $\fctTrmaps{\mappingExpr'}$.
	By Lemma~\ref{lemma:SubGraphs1},
	%As in Note~\ref{note:SubGraphs1},
this means that $\symRDFGraph''\! \subseteq \symRDFGraph'$ and, thus, $t \in \symRDFGraph'$. Consequently, $\mu \in \evalP{\symTP}{\symRDFGraph'\!}$.

To show that $\evalP{\symTP}{\symRDFGraph} \supseteq \evalP{\symTP}{\symRDFGraph'}$, we let $\mu$ be a solution mapping in~$\evalP{\symTP}{\symRDFGraph'}$ and show that $\mu \in \evalP{\symTP}{\symRDFGraph}$:
Since $\mu \in \evalP{\symTP}{\symRDFGraph'}$, we know that there exists a triple $t$ in $\symRDFGraph'$ such that $\mu[\symTP] = t$.
	By Lemma~\ref{lemma:SubGraphs2},
	%As in Note~\ref{note:SubGraphs2},
we have that $\symRDFGraph'\! \subseteq \symRDFGraph$. Therefore, the triple $t$ is also in $\symRDFGraph$ and, thus, $\mu \in \evalP{\symTP}{\symRDFGraph}$.

% \emph{Induction step:} We distinguish the following two cases (while there exist more cases, the proof for them is almost the same as for the cases considered here).
% \begin{itemize}
% 	\item $\symPattern$ is of the form $(\symPattern_1 \OpAND \symPattern_2)$.
% 	In this case, the induction hypothesis is that $\evalP{\symPattern_1}{\symRDFGraph} = \evalP{\symPattern_1}{\symRDFGraph'}$ and $\evalP{\symPattern_2}{\symRDFGraph} = \evalP{\symPattern_2}{\symRDFGraph'}$.
% 	Consequently,
% 		%$\evalP{(\symPattern_1 \OpAND \symPattern_2)}{\symRDFGraph} = \evalP{(\symPattern_1 \OpAND \symPattern_2)}{\symRDFGraph'}$.
% 		$\evalP{\symPattern}{\symRDFGraph} = \evalP{\symPattern}{\symRDFGraph'}$.
% 
% 	\item $\symPattern$ is of the form $(\symPattern_1 \OpOPT \symPattern_2)$.
% 	Also in this case, the induction hypothesis is that $\evalP{\symPattern_1}{\symRDFGraph} = \evalP{\symPattern_1}{\symRDFGraph'}$ and $\evalP{\symPattern_2}{\symRDFGraph} = \evalP{\symPattern_2}{\symRDFGraph'}$. Therefore, $\evalP{\symPattern}{\symRDFGraph} = \evalP{\symPattern}{\symRDFGraph'}$. \qed
% \end{itemize}

\emph{Induction step:} We consider the case that $\symPattern$ is of the form $(\symPattern_1 \OpAND \symPattern_2)$. In this case, the induction hypothesis is that $\evalP{\symPattern_1}{\symRDFGraph} = \evalP{\symPattern_1}{\symRDFGraph'}$ and $\evalP{\symPattern_2}{\symRDFGraph} = \evalP{\symPattern_2}{\symRDFGraph'}$. Consequently, $\evalP{(\symPattern_1 \OpAND \symPattern_2)}{\symRDFGraph} = \evalP{(\symPattern_1 \OpAND \symPattern_2)}{\symRDFGraph'}$. While there exist more cases~(other forms of graph patterns), the proof for them is essentially the same. \qed

\enlargethispage{\baselineskip}%  % Layout Adjustment

\subsection*{Proof of Proposition~\ref{prop:SatisfiabilityForTPsUndecidable}}
\addcontentsline{toc}{section}{Proof of Proposition~\ref{prop:SatisfiabilityForTPsUndecidable}}

We prove
	%the undecidability of \textsc{Satisfiability(TPoverTrMap)} by reduction. In particular, we reduce
	the undecidability of \textsc{Satisfiability(TPoverTrMap)} by reducing
	%Proposition~\ref{prop:SatisfiabilityForTPsUndecidable} by reducing
the satisfiability problem of the relational algebra---which is well know to be undecidable~\cite[Theorem~6.3.1, p.123, together with Theorem~5.3.10, p.80]{DBLP:books/aw/AbiteboulHV95}---to \textsc{Satisfiability(TPoverTrMap)}.

To this end, let
$\symDataObjUni_\textsf{RDB} \subset \symDataObjUni$ be the set of all relational databases (RDBs),
$\symDataObjUni_\textsf{RT} \subset \symDataObjUni$ be the set of all relational tuples, 
$\symDataAcc_\textsf{RA} \in \symDataAccUni$ be the set of all relational algebra expressions, and $\eval_\textsf{RA} \!: \symDataAcc_\textsf{RA} \times \symDataObjUni_\textsf{RDB} \rightarrow 2^{\symDataObjUni_\textsf{RT}}$ be the function that defines the evaluation semantics of the relational algebra. Then, the satisfiability problem of the relational algebra is the following decision problem.

\medskip \noindent
\begin{tabular}{|p{\textwidth}|}
	\hline
	\textbf{Problem:} \textsc{Satisfiability(RA)} \\[0.5mm]
	Input: a relational algebra expression $q \in \symDataAcc_\textsf{RA}$
	\\
	Question: Does there exist an RDB $r \in \symDataObjUni_\textsf{RDB}$ such that the set $\eval_\textsf{RA}(q,r)$ of
	\\ \phantom{Question:} result tuples is not empty?
	%\\
	%Question: Does there exist an $\mathit{rdb} \in \symDataObjUni_\textsf{RDB}$ such that $\eval_\textsf{RA}(q,\mathit{rdb})$ is not empty?
	\\
	\hline
\end{tabular}
\smallskip

For the reduction we need
	a function~$f$ that maps
	%to map
every input for \textsc{Satisfiability(RA)}, i.e., every
	relational algebra expression
$q \in \symDataAcc_\textsf{RA}$, to an input for \textsc{Satisfiability(TPoverTrMap)}, i.e., a triple pattern and a TrMap-ex\-pres\-sion.
To define $f$ we assume
an arbitrary data object $d^*\! \in \symDataObjUni$,
a query language $\symDataAcc_\textsf{x} \in \symDataAccUni$ with a single query $q^*\! \in \symDataAcc_\textsf{x}$,
and an
	%arbitrary
RDF literal $\ell^*\! \in \symAllLiterals$
($d^*\!$, $q^*\!$, and $\ell^*$ do not need to be specified further for the purpose of this proof); and we introduce a source type $\sourceType_\textsf{RAx} = (\symDataObjUni_\textsf{RDB}, \symDataObjUni_\textsf{RT}, \symDataObjUni_\textsf{x}, \symDataAcc_\textsf{RA}, \symDataAcc_\textsf{x}, \eval_\textsf{RA}, \eval_\textsf{x}, \typeCast_\textsf{x})$ where:
\begin{itemize}
	\item
	$\symDataObjUni_\textsf{RDB}$, $\symDataObjUni_\textsf{RT}$, $\symDataAcc_\textsf{RA}$, $\symDataAcc_\textsf{x}$, and $\eval_\textsf{RA}$ are defined as mentioned above;
%	\item
	$\symDataObjUni_\textsf{x} = \{ d^* \}$;
	\item
	$\eval_\textsf{x} \!: \symDataObjUni_\textsf{RDB} \times \symDataObjUni_\textsf{RT} \times \symDataAcc_\textsf{x} \rightarrow \symDataObjUni_\textsf{x}$ is defined such that, for every RDB $r \in \symDataObjUni_\textsf{RDB}$ and every tuple $t \in \symDataObjUni_\textsf{RT}$,
		it holds that
	$\eval_\textsf{x}(r,t,q^*) = \{ d^* \}$ (i.e., $\eval_\textsf{x}$ has the same result for every possible input and, thus, is a constant function);
	\item
	$\typeCast_\textsf{x} \!: \symDataObjUni_\textsf{x} \rightarrow \symAllLiterals$ is defined such that $\typeCast_\textsf{x}(d^*\!) = \ell^*\!$.
\end{itemize}
Now we define
	the function~%
$f$. For every $q \in \symDataAcc_\textsf{RA}$, $f$ maps $q$ to the pair~$(\symTP^*\!,\mappingExpr_q)$ with $\symTP^*$ being the triple pattern $(v,v,v)$ and $\mappingExpr_q$ being the TrMap-ex\-pres\-sion
$
% 	\fctExtOpX{\graphAttr}{\extExpr_\mathrm{g}}{
		\fctExtOpX{\objAttr}{\extExpr_\mathrm{o}}{
			\fctExtOpX{\predAttr}{\extExpr_\mathrm{p}}{
				\fctExtOpX{\subjAttr}{\extExpr_\mathrm{s}}{
					\fctExtractX{\sourceReference}{\sourceType_\textsf{RAx}}{\queryExpr}{\symAttrQueryMap}
				}
			}
		}
% 	}
$
such that
\begin{itemize}
	\item
	$v$ is an arbitrary variable (i.e., $v \in \symAllVariables$);
% 	\item
% 	$\extExpr_\mathrm{g}$ is the IRI \ttl{rr:defaultGraph};
	\item
	$\extExpr_\mathrm{o}$ is an arbitrary IRI $\symIRI \in \symAllIRIs$;
	\item
	$\extExpr_\mathrm{p}$ and $\extExpr_\mathrm{s}$ are also the IRI $\symIRI$, respectively;
	\item
	$\sourceReference$ is an arbitrary source reference (i.e., $\sourceReference \in \symSourcerefUni$);
	\item
	$q$ is the given relational algebra expression; and
	\item
	the partial function $\symAttrQueryMap \!: \symAttrUniverse \rightarrow \symDataAcc_\textsf{x}$ is defined such that $\fctDom{\symAttrQueryMap} = \{ \attr^* \}$ and $\symAttrQueryMap(\attr^*) = q^*\!$, where $\attr^*$ is an arbitrary attribute in $\symAttrUniverse \setminus \{ \subjAttr, \predAttr, \objAttr%
		%, \graphAttr
	\}$.
\end{itemize}
	Notice
	%We emphasize
that, for every $q \in \symDataAcc_\textsf{RA}$, the
	%resulting TrMap-ex\-pres\-sion~%
	TrMap-ex\-pres\-sion~%
	%resulting
$\mappingExpr_q$ has the following properties: For every source assignment~$\validInput$ that is valid input for $\mappingExpr_q$ it must hold that $\sourceReference \in \fctDom{\validInput}$ and $\validInput(\sourceReference) \in \symDataObjUni_\textsf{RDB}$.
%
%	Given such a source assignment~$\validInput$, it holds that the mapping relation $\fctApply{\mappingExpr_q}{\validInput}$ is not empty if and only if $\eval_\textsf{RA}\bigl(q, \validInput(\sourceReference) \bigr)$ is not empty.
%
Given such a source assignment~$\validInput$ and the assigned RDB $\mathit{rdb} = \validInput(\sourceReference)$, for every (relational) tuple $t \in \eval_\textsf{RA}(q, \mathit{rdb})$, the $\extract$ operator in $\mappingExpr_q$ creates mapping tuples for attribute $\attr^*\!$, with values obtained by $\typeCast_\textsf{x}(d)$ for every data object $d \in \eval_\textsf{x}(\mathit{rdb},t,q^*)$. Yet, since
	$\eval_\textsf{x}$ is a constant function with
$\eval_\textsf{x}(\mathit{rdb},t,q^*) = \{ d^* \}$, every $t \in \eval_\textsf{RA}(q, \mathit{rdb})$ is mapped to the same mapping tuple~$t'$ with $t'(\attr^*) = \typeCast_\textsf{x}(d^*) = \ell^*\!$. Therefore, if $\eval_\textsf{RA}(q, \mathit{rdb})$ is not empty, then the final mapping relation $\fctApply{\mappingExpr_q}{\validInput}$ contains a single mapping tuple: $t''\! = \{ \attr^*\! \rightarrow \ell^*\!, \subjAttr \rightarrow \symIRI, \predAttr \rightarrow \symIRI, \objAttr \rightarrow \symIRI \}$, and the RDF graph resulting from that mapping relation contains a single triple: $(\symIRI,\symIRI,\symIRI)$. The triple pattern $\symTP^*\! = (v,v,v)$ matches this triple. In contrast, if $\eval_\textsf{RA}(q, \mathit{rdb})$ is empty, then
	%the mapping relation
$\fctApply{\mappingExpr_q}{\validInput}$ is empty and, thus, the resulting RDF graph is empty%
	%; i.e., nothing to match for $\symTP^*\!$%
.

Based on these observations, we can conclude that, for every
	%relational algebra expression
$q \in \symDataAcc_\textsf{RA}$ and every $\mathit{rdb} \in \symDataObjUni_\textsf{RDB}$, it holds that $\eval_\textsf{RA}(q, \mathit{rdb}) \neq \emptyset$ if and only if $\evalP{\symTP^*}{\symRDFGraph} \neq \emptyset$, where $\symRDFGraph$ is the RDF graph resulting from the mapping relation $\fctApply{\mappingExpr_q}{\validInput_r}$ for which $\validInput_r$ is an arbitrary source assignment with
	 $\sourceReference \in \fctDom{\validInput_\mathit{rdb}}$ and
$\validInput_\mathit{rdb}(\sourceReference) = \mathit{rdb}$.
As a consequence, for every
	%relational algebra expression
$q \in \symDataAcc_\textsf{RA}$, it holds that there exists an $\mathit{rdb} \in \symDataObjUni_\textsf{RDB}$ such that $\eval_\textsf{RA}(q, \mathit{rdb}) \neq \emptyset$ if and only if $\symTP^*$ is satisfiable over $\mappingExpr_q$.

Therefore, if
	we assume that \textsc{Satisfiability(TPoverTrMap)} is
	%\textsc{Satisfiability(TPoverTrMap)} was
decidable,
	we could use the decider for it to also decide \textsc{Satisfiability(RA)}. Yet, since \textsc{Satisfiability(RA)} is undecidable~\cite{DBLP:books/aw/AbiteboulHV95}, we would have a contradiction and, thus, \textsc{Satisfiability(TPoverTrMap)} cannot be decidable.
	%then \textsc{Satisfiability(RA)} would be decidable, which is not the case~\cite{DBLP:books/aw/AbiteboulHV95}.
\qed

\subsection*{Proof of Proposition~\ref{prop:IncompatibilityImpliesUnsatisfiability}}
\addcontentsline{toc}{section}{Proof of Proposition~\ref{prop:IncompatibilityImpliesUnsatisfiability}}

Let $\symTP = (s,p,o)$ be a triple pattern that is incompatible with $\mappingExpr$. To show that $\symTP$ is not satisfiable over~$\mappingExpr$, we assume an arbitrary source assignment~$\validInput$ that is a valid input for $\mappingExpr$ and show, without loss of generality, that $\evalP{\symTP}{\symRDFGraph} = \emptyset$ where $\symRDFGraph$ is the RDF graph resulting from the mapping relation $\fctApply{\mappingExpr}{\validInput} = \mappingRel$.

First, we consider the case that $\mappingExpr$ is of the first of the two forms given in Definition~\ref{def:TrMapExpression}. In this case,
	%by Definition~\ref{def:IncompatibleTP},
we know that at least one of the conditions \ref{case:IncompatibleTP:1:1}--\ref{case:IncompatibleTP:1:6} in Definition~\ref{def:IncompatibleTP} holds for $\symTP$ and $\mappingExpr$. For each of these conditions, we now show that, given the condition, it holds that $\evalP{\symTP}{\symRDFGraph} = \emptyset$.
\begin{itemize}
	\item Suppose $s$ is an IRI that is incompatible with $\extExpr_\mathrm{s}$ (i.e., condition~\ref{case:IncompatibleTP:1:1}). In this case, for every mapping tuple~$\mappingTuple \in \symMappingInst$, it holds that $\fctEval{\extExpr_\mathrm{s}}{\mappingTuple} \neq s$%
		%, which follows readily from Definitions~\ref{def:RegEx} and~\ref{def:IncompatibleIRI}, in combination with Definition~\ref{def:ExtendExpressions:Semantics} and the definition of the extension functions being used~($\toIRI$, $\texttt{toBNode}$, $\toLiteral$, and indirectly $\concat$; all defined in \cite[Appendix~B]{minoo2025AlgebraPublished}). From $\fctEval{\extExpr_\mathrm{s}}{\mappingTuple} \neq s$ follows that $\mappingTuple(\subjAttr) \neq s$, which means that
		, which follows readily from Definitions~\ref{def:RegEx} and~\ref{def:IncompatibleIRI}, in combination with the definition of $\fctEval{\extExpr_\mathrm{s}}{\mappingTuple}$~\cite[Def.9]{minoo2025AlgebraPublished} and of the extension functions being used~($\toIRI$, $\texttt{toBNode}$, $\toLiteral$, and indirectly $\concat$; all defined in \cite[Appendix~B]{minoo2025AlgebraPublished}). From $\fctEval{\extExpr_\mathrm{s}}{\mappingTuple} \neq s$ follows that $\mappingTuple(\subjAttr) \neq s$, which means that
		%~(see Lemma~\ref{lemma:IncompatibleIRI}) and, thus, $\mappingTuple(\subjAttr) \neq s$. As a consequence,
	the RDF graph~$\symRDFGraph$ does not contain a triple $(s'\!, p'\!, o')$
		%such that
		s.t.\
	$s'\! = s$ and, thus, $\evalP{\symTP}{\symRDFGraph} = \emptyset$.

% 	\item Suppose $p$ is an IRI that is incompatible with $\extExpr_\mathrm{p}$ (i.e., condition~\ref{case:IncompatibleTP:1:2}). In this case, for every mapping tuple~$\mappingTuple \in \symMappingInst$, it holds that $\fctEval{\extExpr_\mathrm{p}}{\mappingTuple} \neq p$ (see Lemma~\ref{lemma:IncompatibleIRI}) and, thus, $\mappingTuple(\predAttr) \neq p$. As a consequence, the RDF graph~$\symRDFGraph$ does not contain a triple $(s'\!, p'\!, o')$ such that $p'\! = p$ and, thus, $\evalP{\symTP}{\symRDFGraph} = \emptyset$.
% 
% 	\item Suppose $o$ is an IRI that is incompatible with $\extExpr_\mathrm{o}$ (i.e., condition~\ref{case:IncompatibleTP:1:3}). In this case, for every mapping tuple~$\mappingTuple \in \symMappingInst$, it holds that $\fctEval{\extExpr_\mathrm{o}}{\mappingTuple} \neq o$ (see Lemma~\ref{lemma:IncompatibleIRI}) and, thus, $\mappingTuple(\objAttr) \neq o$. As a consequence, the RDF graph~$\symRDFGraph$ does not contain a triple $(s'\!, p'\!, o')$ such that $o'\! = o$ and, thus, $\evalP{\symTP}{\symRDFGraph} = \emptyset$.

\item Suppose $p$ is an IRI that is incompatible with $\extExpr_\mathrm{p}$ (i.e., condition~\ref{case:IncompatibleTP:1:2}) or $o$ is an IRI that is incompatible with $\extExpr_\mathrm{o}$ (i.e., condition~\ref{case:IncompatibleTP:1:3}). For
	%both of
these cases, $\evalP{\symTP}{\symRDFGraph} = \emptyset$ can be shown
	%using the same argument as used
	in the same way as done
for the previous~case.

	\item Suppose $o$ is a literal and $\extExpr_\mathrm{o}$ is
		%either
	of the form $\fctToIRI{\extExpr'}{\IbaseIRI}$ or of the form $\fctToBNode{\extExpr'}$ (%
		i.e.,
	conditions~\ref{case:IncompatibleTP:1:4} and~\ref{case:IncompatibleTP:1:5}). In this case, for every mapping tuple~$\mappingTuple \in \symMappingInst$, $\fctEval{\extExpr_\mathrm{o}}{\mappingTuple}$ is not a literal (by the definition of the $\toIRI$ function and the $\texttt{toBNode}$ function) and, thus, $\fctEval{\extExpr_\mathrm{o}}{\mappingTuple} \neq o$ and $\mappingTuple(\objAttr) \neq o$.
		%As a consequence, the RDF graph~$\symRDFGraph$ does not contain a triple $(s'\!, p'\!, o')$ such that
		Consequently, there is no triple $(s'\!, p'\!, o')$ in $\symRDFGraph$ s.t.\
	$o'\! = o$ and, thus, $\evalP{\symTP}{\symRDFGraph} = \emptyset$.

	\item Suppose $o$ is a literal $\literalTuple$ and $\extExpr_\mathrm{o}$ is of the form $\fctToLiteral{\extExpr'\!, \dt'}$ such that $\lex$ does not match the regular expression $\mathrm{regex}(\extExpr')$ or $\dt \neq \dt'\!$~(%
		i.e.,
	condition~\ref{case:IncompatibleTP:1:6}). In this case, for every mapping tuple~$\mappingTuple \in \symMappingInst$, $\fctEval{\extExpr_\mathrm{o}}{\mappingTuple}$ may be either the error symbol, $\error$, or a literal $(\lex''\!, \dt'')$ with $\dt''\! = \dt'\!$ (which follows from the definition of the $\toLiteral$ function~\cite{minoo2025AlgebraPublished}). In the latter case, if $\lex$ does not match the regular expression $\mathrm{regex}(\extExpr')$, it must hold that $\lex \neq \lex''$ and, thus, $\literalTuple \neq (\lex''\!, \dt'')$. Likewise, if $\dt \neq \dt'\!$, it also holds that $\literalTuple \neq (\lex''\!, \dt'')$. Hence, in all cases, we have that $\fctEval{\extExpr_\mathrm{o}}{\mappingTuple} \neq o$ and, thus, $\mappingTuple(\objAttr) \neq o$%
		%. As a consequence, the RDF graph~$\symRDFGraph$ does not contain a triple $(s'\!, p'\!, o')$ such that $o'\! = o$ and, thus, $\evalP{\symTP}{\symRDFGraph} = \emptyset$.
		, which leads to $\evalP{\symTP}{\symRDFGraph} = \emptyset$, as in the previous cases.
\end{itemize}

Now we consider the case that $\mappingExpr$ is of the second of the two forms given in Definition~\ref{def:TrMapExpression}. In this case,
	%by Definition~\ref{def:IncompatibleTP},
we have that at least one of the conditions~\ref{case:IncompatibleTP:2:1}--\ref{case:IncompatibleTP:2:4} in Definition~\ref{def:IncompatibleTP} holds for $\symTP$ and $\mappingExpr$. For conditions~\ref{case:IncompatibleTP:2:1}--\ref{case:IncompatibleTP:2:3}, we can show that $\evalP{\symTP}{\symRDFGraph} = \emptyset$ by using the same argument as
	%we have
used above for conditions~\ref{case:IncompatibleTP:1:1}--\ref{case:IncompatibleTP:1:3}.
It remains to discuss condition~\ref{case:IncompatibleTP:2:4}, for which we assume that $o$ is a literal and we notice that, for every mapping tuple~$\mappingTuple \in \symMappingInst$, $\fctEval{\extExpr_\mathrm{o}'}{\mappingTuple}$ is not a literal. The latter follows from the fact that, by Definition~\ref{def:TrMapExpression}, $\extExpr_\mathrm{o}'$ is either an IRI, a blank node, or of the form $\fctToIRI{\extExpr'}{\IbaseIRI}$ or $\fctToBNode{\extExpr'}$. Therefore, we have that $\fctEval{\extExpr_\mathrm{o}'}{\mappingTuple} \neq o$ and, thus, $\mappingTuple(\objAttr) \neq o$, which
	%leads to $\evalP{\symTP}{\symRDFGraph} = \emptyset$, as in the previous cases.
	again leads to $\evalP{\symTP}{\symRDFGraph} = \emptyset$.
\qed

}

\end{document}